\documentclass[final,twoside,11pt]{entics} 
\usepackage{enticsmacro}
\usepackage{graphicx}
\usepackage[all]{xy}

\usepackage{amsmath,stmaryrd,textcomp,amssymb,enumerate}
\usepackage{algorithm} 
\usepackage{mathtools}
\usepackage{hyperref}
\usepackage{color}
\usepackage{relsize}
\usepackage{bm}
\usepackage{bbm} 
\usepackage{subcaption} 
\usepackage{cleveref}
\usepackage{ebproof}
\usepackage{pifont}
\usepackage{wrapfig}
\usepackage{cutwin}
\usepackage{tikz-cd}
\usepackage{tikzit}

\tikzstyle{red dot}=[fill=red, draw=black, shape=circle, scale=0.3]
\tikzstyle{medium box}=[fill=white, draw=black, shape=rectangle, minimum width=0.3cm, minimum height=0.5cm]
\tikzstyle{s flat}=[fill=white, draw=black, shape=rectangle, minimum width=8mm, minimum height=5mm]
\tikzstyle{black dot}=[fill=black, draw=black, shape=circle, scale=0.3]
\tikzstyle{empty dot}=[fill=none, draw=black, shape=circle, scale=0.3]
\tikzstyle{l flat}=[fill=white, draw=black, shape=rectangle, minimum width=1.8cm, minimum height=0.3cm]
\tikzstyle{s rect}=[fill=white, draw=black, shape=rectangle, minimum width=0.1cm, minimum height=0.1cm]
\tikzstyle{s vert}=[fill=white, draw=black, shape=rectangle, minimum width=5mm, minimum height=8mm]
\tikzstyle{m vert}=[fill=white, draw=black, shape=rectangle, minimum width=5mm, minimum height=12mm]
\tikzstyle{m flat}=[fill=white, draw=black, shape=rectangle, minimum width=6mm, minimum height=5mm]
\tikzstyle{mm flat}=[fill=white, draw=black, shape=rectangle, minimum height=5mm, minimum width=10mm]
\tikzstyle{mmm flat}=[fill=white, draw=black, shape=rectangle, minimum height=5mm, minimum width=12mm]
\tikzstyle{grey dot}=[fill={rgb,255: red,191; green,191; blue,191}, draw={rgb,255: red,191; green,191; blue,191}, shape=circle, scale=0.3]
\tikzstyle{mm vert}=[fill=white, draw=black, shape=rectangle, minimum width=5mm, minimum height=14mm]
\tikzstyle{20mm vert}=[fill=white, draw=black, shape=rectangle, minimum width=5mm, minimum height=20mm]
\tikzstyle{16mm vert}=[fill=white, draw=black, shape=rectangle, minimum width=5mm, minimum height=16mm]
\tikzstyle{18mm vert}=[fill=white, draw=black, shape=rectangle, minimum width=5mm, minimum height=18mm]
\tikzstyle{grey s vert}=[fill=white, draw={rgb,255: red,191; green,191; blue,191}, shape=rectangle, minimum width=5mm, minimum height=8mm]
\tikzstyle{grey s rect}=[fill=white, draw={rgb,255: red,191; green,191; blue,191}, shape=rectangle, minimum width=0.1mm, minimum height=0.1mm]

\tikzstyle{dashes}=[-, dashed]
\tikzstyle{right arrow}=[->]
\tikzstyle{left arrow}=[<-]
\tikzstyle{grey fill}=[-, fill={rgb,255: red,191; green,191; blue,191}, draw={rgb,255: red,191; green,191; blue,191}]
\tikzstyle{blue fill}=[-, fill=cyan, draw=cyan]
\tikzstyle{yellow fill}=[-, fill=yellow, draw=yellow]
\tikzstyle{green fill}=[-, fill=green, draw=green]
\tikzstyle{red wire}=[-, draw=red]
\tikzstyle{blue wire}=[-, draw=blue]
\tikzstyle{grey wire}=[-, draw={rgb,255: red,191; green,191; blue,191}, fill=none]
\tikzstyle{green wire}=[-, draw=green]
\tikzstyle{yellow wire}=[-, draw=yellow]
\tikzstyle{black fill}=[-, fill=black]
\tikzstyle{red fill}=[-, fill=red, draw=red]

\newcommand{\copyy}{\input{copy.tikz}}
\newcommand{\discard}{\begin{tikzpicture}
	\begin{pgfonlayer}{nodelayer}
		\node [style=none] (0) at (-0.5, 0) {};
		\node [style=black dot] (1) at (0.5, 0) {};
	\end{pgfonlayer}
	\begin{pgfonlayer}{edgelayer}
		\draw (0.center) to (1);
	\end{pgfonlayer}
\end{tikzpicture}
}
\newcommand{\cocopy}{\input{cocopy.tikz}}
\newcommand{\codiscard}{\begin{tikzpicture}
	\begin{pgfonlayer}{nodelayer}
		\node [style=none] (2) at (0.45, 0) {};
		\node [style=black dot] (4) at (-0.55, 0) {};
	\end{pgfonlayer}
	\begin{pgfonlayer}{edgelayer}
		\draw (4) to (2.center);
	\end{pgfonlayer}
\end{tikzpicture}
}
\newcommand{\conjj}{\input{conj.tikz}}
\newcommand{\unit}{\begin{tikzpicture}
	\begin{pgfonlayer}{nodelayer}
		\node [style=none] (0) at (0.4, 0) {};
		\node [style=empty dot] (2) at (-0.6, 0) {};
	\end{pgfonlayer}
	\begin{pgfonlayer}{edgelayer}
		\draw (2) to (0.center);
	\end{pgfonlayer}
\end{tikzpicture}
}
\newcommand{\coconj}{\input{coconj.tikz}}
\newcommand{\counit}{\begin{tikzpicture}
	\begin{pgfonlayer}{nodelayer}
		\node [style=none] (0) at (-0.5, 0) {};
		\node [style=empty dot] (1) at (0.5, 0) {};
	\end{pgfonlayer}
	\begin{pgfonlayer}{edgelayer}
		\draw (0.center) to (1);
	\end{pgfonlayer}
\end{tikzpicture}
}

\newcommand{\Bool}{\mathbb{B}}

    \newcommand{\setTwo}{\{0, 1\}}

\newcommand{\At}{\mathit{At}}

\newcommand{\axBARel}{\mathit{SATA}}
\newcommand{\catSata}{\mathsf{SATA}}

\newcommand{\bigcdot}{\ensuremath{%
   \mathchoice%
    {\mskip\thinmuskip\lower0.2ex\hbox{\scalebox{1.5}{$\cdot$}}\mskip\thinmuskip}}%
    {\mskip\thinmuskip\lower0.2ex\hbox{\scalebox{1.5}{$\cdot$}}\mskip\thinmuskip}%
    {\lower0.3ex\hbox{\scalebox{1.2}{$\cdot$}}}%
    {\lower0.3ex\hbox{\scalebox{1.2}{$\cdot$}}}%
}

\newcommand{\body}{\mathsf{body}}

\newcommand{\catBProf}{\mathsf{MonRel}}

\newcommand{\catRel}{\mathbf{Rel}}

\newcommand{\catSyn}{\mathsf{Syn}}
    
\newcommand{\catSynDef}{\mathsf{Syn}}

\newcommand{\clauseC}{\varphi}
\newcommand{\clauseD}{\psi}

\newcommand{\cnsqDiag}[1]{\mathsf{C}_{#1}}
    \newcommand{\cnsqOp}[1]{\mathbf{C}_{#1}}

\newcommand{\diag}[1]{\mathsf{diag}(#1)}

\newcommand{\empset}{\varnothing}

\newcommand{\funcOne}{\mathbf{1}}
    \newcommand{\funcZero}{\mathbf{0}}

\newcommand{\head}{\mathsf{head}}

\newcommand{\imCnsqDiag}[1]{\mathsf{T}_{#1}}
    \newcommand{\imCnsqOp}[1]{\mathbf{T}_{#1}}

\newcommand{\intI}{\mathcal{I}}

	\newcommand{\intJ}{\mathcal{J}}

    \newcommand{\intUndirProf}[1]{\langle #1 \rangle}
    \newcommand{\intBProf}[1]{\left\langle #1 \right\rangle}

\newcommand{\iso}{\simeq}

    \newcommand{\leastHMod}[1]{\mathbf{M}_{#1}}

\newcommand{\lr}[1]{\langle #1 \rangle}

\newcommand{\matEntry}[2]{\delta_{#1 #2}}

\newcommand{\negClause}{$\neg$-clause}
    \newcommand{\negClauses}{$\neg$-clauses}

\newcommand{\obLeft}[1]{{#1}^{\blacktriangleleft}}
    \newcommand{\obRight}[1]{{#1}^{\blacktriangleright}}

\newcommand{\progL}{\mathbb{L}}

\newcommand{\progP}{\mathbb{P}}

\newcommand{\progQ}{\mathbb{Q}}

\newcommand{\rowR}{r}

\newcommand{\swap}{\input{swap.tikz}}
\newcommand{\sysIneq}[2]{\mathsf{clause}_{#1}v(#2)}

\newcommand{\tensor}{\oplus}

\newcommand{\then}{\,;}
    \newcommand{\relThen}{\mathbin{\fatsemi}}

\newcommand{\Diagbox}[3]{
\begin{tikzpicture}[scale = 1, baseline = -0.5pt]
	\begin{pgfonlayer}{nodelayer}
		\node [style=s rect] (0) at (0, 0) {\scriptsize $#1$};
		\node [style=none] (1) at (-0.5, 0) {};
		\node [style=none] (2) at (0.5, 0) {};
		\node [style=none] (3) at (-0.5, 0.2) {\scriptsize $#2$};
		\node [style=none] (5) at (0.5, 0.2) {\scriptsize $#3$};
	\end{pgfonlayer}
	\begin{pgfonlayer}{edgelayer}
		\draw (1.center) to (0);
		\draw (0) to (2.center);
	\end{pgfonlayer}
\end{tikzpicture}}

\newcommand{\valS}{\mathbf{s}}

\newcommand{\Var}{\mathit{Var}}

\makeatletter
\newcommand*{\inlineequation}[2][]{%
  \begingroup
    \refstepcounter{equation}%
    \ifx\\#1\\%
    \else
      \label{#1}%
    \fi
    \relpenalty=10000 %
    \binoppenalty=10000 %
    \ensuremath{%
      #2%
    }%
    ~\@eqnnum
  \endgroup
}
\makeatother

\def\conf{MFPS 2022} 	
\volume{1}			
\def\volu{1}			
\def\papno{4}			
\def\mydoi{{\normalshape \href{https://doi.org/10.46298/entics.10481}{10.46298/entics.10481}}}
\def\copynames{T. Gu, R. Piedeleu, F. Zanasi} 

\def\runtitle{A Complete Diagrammatic Calculus...}
\def\lastname{Gu, Piedeleu, and Zanasi}
\begin{document}
\begin{frontmatter}
  \title{A Complete Diagrammatic Calculus \\ for Boolean Satisfiability} 
  \author{Tao Gu\thanksref{tg}}	
   \thanks[tg]{Email: \href{tao.gu.18@ucl.ac.uk}{\texttt{\normalshape tao.gu.18@ucl.ac.uk}}}
  \author{Robin Piedeleu\thanksref{rp}}	
  \thanks[rp]{Email: \href{r.piedeleu@ucl.ac.uk}{\texttt{\normalshape r.piedeleu@ucl.ac.uk}}}
  \author{Fabio Zanasi\thanksref{fz}}	
    \thanks[fz]{Email: \href{f.zanasi@ucl.ac.uk}{\texttt{\normalshape f.zanasi@ucl.ac.uk}}}
  \address{University College London}  		

\begin{abstract} 
We propose a calculus of string diagrams to reason about satisfiability of Boolean formulas, and prove it to be sound and complete. We then showcase our calculus in a few case studies. First, we consider SAT-solving. Second, we consider Horn clauses, which leads us to a new decision method for propositional logic programs equivalence under Herbrand model semantics.
\end{abstract}
\begin{keyword}String diagrams, Categorical semantics, Boolean algebra, Satisfiability, Logic programming\end{keyword}
\end{frontmatter}

\section{Introduction}\label{sec:intro}

 
In the 1970s, Cook and Levin proved independently~\cite{cook1971complexity,levin1973universal} that a number of difficult problems reduce to that of determining the existence of a satisfying assignment to a given Boolean formula.  For this reason, in spite of its (suspected) theoretical intractability, and perhaps because of its surprising practical feasibility, the problem of Boolean satisfiability has acquired a central importance in logic and computer science.

This paper proposes a new perspective on the algebra of Boolean satisfiability, in the form of a sound and complete calculus of string diagrams.

\smallskip

At first glance, it may appear we already have a satisfactory algebraic treatment of Boolean satisfiability, dating back to the landmark works of George Boole in the 19th century~\cite{boole1854investigation}. Thanks to Boole, we know that classical propositional logic can be formulated as an \emph{algebraic theory}, whose models are what we now call Boolean algebras. However, this standard treatment has a catch, which is the starting point of our work: \emph{the satisfiability of a propositional formula cannot be stated as a sentence in the algebraic theory of Boolean algebras}. The reason is simple -- algebraic theories only allow axioms as equations between terms with free variables (or, equivalently from the point of view of first-order classical provability, as equations in which every variable is universally quantified). On the other hand, a Boolean formula $f$ containing free variables $x_1,\dots, x_n$, is satisfiable if and only if
$\exists x_1\dots\exists x_n (f = 1)$,
which is outside of the algebraic realm. Call this formula SAT($f$); the algebraic theory of Boolean algebras is insufficiently expressive to encode SAT($f$) as a statement in the theory itself. 

Of course, it is a perfectly well-formed first-order statement, so we could just use first-order logic to reason about SAT and derive the (un)satisfiability of particular instances. 
However, we argue that a genuinely algebraic treatment of SAT is possible and preferable, provided that we move to a \emph{diagrammatic} syntax. To this effect, we will introduce a calculus to compose systems of Boolean constraints expressed as string diagrams, and reason about them in a purely equational way. 

Why would we need a \emph{diagrammatic} syntax? As we have just shown, the algebraic theory of Boolean algebras is insufficiently expressive, thus we are looking for some other formal system that is more closely tailored to satisfiability. Existential quantification gives the problem a particularly relational flavour: SAT($f$) does not involve evaluating a Boolean function, but checking that there exists some assignment for which the function $f$ evaluates to true. This is a fundamentally \emph{relational}, not functional constraint. And diagrammatic calculi are particularly well suited to express relational constraints, as demonstrated by the wealth of diagrammatic languages developed in recent years for different subcategories of relations (linear~\cite{zanasi2015interacting,bonchi2017interacting}, polyhedral~\cite{bonchi2021diagrammatic}, affine~\cite{bonchi2019graphical}, piecewise-linear~\cite{boisseau2022graphical}). In this setting, string diagrams have the advantage of highlighting key structural features, such as dependencies and connectivity between different sub-terms/diagrams. Finally, another related way to think about string diagrams is as a multi-sorted algebraic syntax,  generalising standard algebraic syntax to the regular fragment of first-order logic, \emph{i.e.} the fragment containing truth, conjunction, and existential quantification~\cite{bonchi2018gcq,bonchi2017functorial}. 

In order to develop our approach, we take as semantic building blocks certain sets of satisfying assignments of Boolean formulas, which we organise into a symmetric monoidal category. We then devise a diagrammatic syntax for them generated by
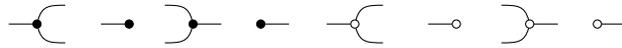
\begin{equation*}
    \input{generate-ar.tikz}
\end{equation*}
A diagram in this syntax is formed by composing any of the generators horizontally or vertically, connecting and crossing wires where needed. A dangling wire represents a free variable so that a diagram with $m$ wires on the left and $n$ wires on the right can be thought of as a Boolean formula in CNF with $m+n$ free variables. The intuition is that variables corresponding to wires on the left appear negatively in the associated formula, while those on the right appear positively. The correspondence between the syntax and the Boolean semantics is simple: each diagram is interpreted as the set of satisfying assignments of an associated formula.  It is helpful to give these formulas in CNF for the generators explicitly:
\begin{align*}
\copyy &\mapsto (\lnot x\lor y_1) \land (\lnot x\lor y_2)  \qquad &\discard\mapsto 1 \qquad\cocopy &\mapsto (\lnot x_1\lor y) \land (\lnot x_2\lor y)\qquad &\codiscard &\mapsto 1 \\
 \coconj &\mapsto \lnot x\lor y_1\lor y_2\qquad &\unit\mapsto y\qquad  \conjj &\mapsto \lnot x_1\lor \lnot x_2\lor y\qquad &\counit &\mapsto \lnot x
\end{align*}
with the convention that $x$s correspond to left wires and $y$s to right wires. 
Then, composing two diagrams in parallel (vertically) amounts to taking the conjunction of their associated formulas, while composing them in series (horizontally) requires existentially quantifying over the variable corresponding to the shared wire, effectively projecting it out.
\begin{wrapfigure}{r}{0.25\textwidth}
    \input{unsat-instance-ex.tikz}
\end{wrapfigure} Using these, we are able to encode arbitrary SAT instances -- statements of the form $\exists x f$ for some CNF formula $f$ -- as diagrams without any dangling wires.  For example, the SAT instance corresponding to $\exists x\exists y(\lnot x\lor y)\land (x\lor y)\land (x\lor \lnot y)\land (\lnot x\lor \lnot y)$ can be depicted as on the right.

The main technical result of this paper is a sound and complete equational theory for the intended semantics: any (in)equality that holds in the Boolean semantics, can be derived by equational reasoning at the level of the diagrams themselves. This theory, we argue, is the \emph{algebra of satisfiability}. 

Using our complete axiomatisation, we can derive the (un)satisfiability of any given SAT instance entirely diagrammatically, by applying local equational reasoning steps. For example, in a few steps, diagram above can be rewritten to $\begin{tikzpicture}
	\begin{pgfonlayer}{nodelayer}
		\node [style=empty dot] (41) at (0.25, 0) {};
		\node [style=empty dot] (42) at (-0.25, 0) {};
	\end{pgfonlayer}
	\begin{pgfonlayer}{edgelayer}
		\draw (42) to (41);
	\end{pgfonlayer}
\end{tikzpicture}
$, a diagram which represents the unsatisfiable formula $\exists x(\lnot x\land x)$. 

We also take advantage of the fact that our syntax can represent \emph{Horn clauses} to give a sound and complete encoding of propositional logic programs. We show how to represent definite logic programs diagrammatically and how their conventional semantics 
corresponds to the semantics of the associated diagrams. Paired to the completeness result, this gives us a procedure to decide semantic equivalence of logic programs, by rewriting in the equational theory of our calculus -- see e.g. Example~\ref{ex:lpequiv} below.

\noindent\textbf{Synopsis.} We introduce and motivate the semantic domain of interest in Section~\ref{sec:semantics}. We introduce the diagrammatic calculus in Section~\ref{sec:cat-sem}. We show how to use the calculus to reason about satisfiability in Section~\ref{sec:sat}, and how to use it as a proof system for equivalence of logic programs in Section~\ref{sec:LP}. We conclude with the proof of completeness in Section~\ref{sec:sound-complete}. 
Due to space limitation, we omit the proof details, which can be found at \href{https://arxiv.org/abs/2211.12629}{arXiv:2211.12629 [cs.LO]}. 

\section{Boolean formulas, Satisfiability and Monotone Relations}
\label{sec:semantics}

In this section we review the basics of Boolean formulas, and justify the introduction of monotone relations as a means to reason about their satisfiability. 

We write $\Bool = \{0, 1\}$ for the two-element Boolean algebra, with the standard order $\{(0,0), (0,1), (1,1)\}$, and $\Bool^{op}$ for the same set with the opposite order. A \emph{Boolean formula} is either $0$, $1$, a variable (from some given set of variables $x,y,\dots$), a conjunction of Boolean formulas $f_1\land f_2$, a disjunction of Boolean formulas $f_1\lor f_2$, or the negation of a Boolean formula $\lnot f$. Any Boolean formula with $n$ free variables admits a standard interpretation as a map $\Bool^n \to \Bool$ given by the usual Boolean semantics of conjunction, disjunction and negation (via their truth-tables). As is common practice, we will use the same notation to denote a Boolean formula $f$ and its interpretation $f\colon\Bool^n \to \Bool$. 

We call \emph{literal} a variable or the negation of a variable. We say that a formula is in \emph{conjunctive normal form} (also known as clausal form) (CNF) when it is given as a conjunction of disjunctions of literals, \emph{e.g.} $(\lnot x \lor y\lor z) \land (x\lor z \lor \lnot y)$. Each of the disjunctions of literals is called a clause. De Morgan duality guarantees that every Boolean formula admits a CNF formula with the same interpretation, \emph{i.e.} that defines the same map $\Bool^n \to \Bool$. The equivalent formula in CNF can be obtained by pushing all the negations to the level of variables and distributing all disjunctions over conjunctions. 

Given a Boolean formula $f: \Bool^n\to \Bool$, we call a \emph{satisfying assignment} of $f$ any tuple of Booleans $(b_1, \dots, b_n)\in \Bool^n$ such that $f(b_1, \dots, b_n) = 1$. We say that a formula is \emph{satisfiable} if it has at least one satisfying assignment. The satisfying assignments of a given formula $f$ with $n$ variables carve out a subset $[f=1] :=\{(b_1,\dots, b_n) : f(b_1,\dots, b_n) = 1\}$ of $\Bool^n$. 

As mentioned in the introduction, satisfiability -- the existence of a satisfying assignment for a given Boolean formula -- is a relational constraint. Furthermore, relations tend to admit well-behaved diagrammatic calculi. For these two reasons, we turn to relations as our semantics.
The class of relations to consider is naturally suggested to us once we examine more closely the satisfying assignments of clauses, the fundamental building block of CNF formulas. Let $c_{m,n} = \lnot x_1\lor \dots \lor \lnot x_m\lor y_1\lor \dots \lor y_n$ be a clause, with $m$ negative literals and $n$ positive ones. Note the satisfying assignments of $c_{m,n}$ form a subset $[c_{m,n}=1]$ of $\Bool^m\times\Bool^n$ that is compatible with the order over $\Bool$: if $(a_1,\dots, a_m, b_1,\dots, b_n)$ is a satisfying assignment of $c_{m,n}$, any other assignment $(a_1',\dots, a_m', b_1',\dots, b_n')$ such that $a_i'\leq a_i$ and $b_j\leq b_j'$ for all $1\leq i\leq m$ and $1\leq j\leq n$, is also satisfying. In other words, $[c_{m,n}=1]$ is an \emph{upward-closed} subset of $(\Bool^{op})^m\times \Bool^n$, where the order on the product of posets $X\times Y$ is given by $(x,y)\leq_{X\times Y} (x',y')$ iff $x\leq_X x'$ and $y\leq_Y y'$. For example, the clause $\lnot x\lor y_1\lor y_2$ has as satisfying assignments $(0,(0,0)), (1, (1,0)), (1,(0,1))$ and all tuples that are greater than these in $\Bool^{op}\times \Bool^2$. Another example: the set of satisfying assignments of $\lnot x \lor y$ is $\{(0,0), (0,1), (1,1)\}$ defines an upward-closed subset of $\Bool^{op}\times\Bool$ which we have already encountered as the order on $\Bool$ itself, \emph{i.e.}, $\{(x,y) \,|\, x\leq y\}$.

This monotonicity property of clauses and their sets of satisfying assignments suggests the following definition for the kind of relations that we wish to study.
\begin{definition}\label{def:monotone-relation}
A \emph{monotone relation} $R:X\to Y$ between two posets $\lr{X, \leq_X}$ and $\lr{Y, \leq_Y}$ is a relation $R \subseteq X \times Y$ satisfying the following monotonicity condition: if $(x, y) \in R$, then for arbitrary $x' \in X$ and $y' \in Y$, $x' \leq x$ and $y \leq y'$ imply $(x', y') \in R$. 
\end{definition}
Monotone relations compose as expected: $R \relThen S \coloneqq \{ (x, z) \mid \exists y \text{ such that } (x, y) \in R, (y, z) \in S \}$. Also, the composition of monotone relations is monotone, so we can organise them into a category.
\begin{proposition}[see e.g.~\cite{fong2019invitation}]
\label{def:monotone-relations-category}
Posets together with monotone relations between them form a category, denoted as $\catBProf$. Moreover, $\catBProf$ equipped with the product of posets and the one-element poset $\lr{\{*\}, \{ (*, *) \}}$ forms a symmetric monoidal category 
\end{proposition}
Note that we can turn a monotone function $f$ into the monotone relation $\{(x,y)\, |\, f(x)\leq y\}\subseteq X\times Y$, so that $\catBProf$ contains the category of monotone functions as a subcategory (\Cref{sec:mono-func-mono-rel}). 

Having established their connection to satisfiability, we focus on monotone relations between tuples of Booleans, \emph{i.e.} between products of the two-element Boolean algebra $\Bool$ (and its dual $\Bool^{op}$), as our primary semantic objects of interest. In the next two sections, we will develop a diagrammatic calculus axiomatising the full symmetric monoidal subcategory of $\catBProf$ spanned by finite Boolean algebras. Moreover, the calculus will shed further light on the relationship between monotonicity and satisfiability. At first glance, from the point of view of satisfiability, it looks like something is lost by restricting our semantics to monotone relations. However, the results of Sections~\ref{sec:sat} and~\ref{sec:LP} demonstrate that, perhaps surprisingly, this is not the case. We will come back to this specific point in Remark~\ref{rmk:monotone-sat}.

\section{The Diagrammatic Calculus}
\label{sec:cat-sem}

In this section we introduce $\catSata$ (\textsc{SAT}isfiability \textsc{A}lgebra), a diagrammatic calculus for monotone relations. It will consist of a syntax of string diagrams (Definition~\ref{def:undir-syn-cat}) and an inequational theory (Figure~\ref{fig:ineq-axiom}). The semantic interpretation will be given by a symmetric monoidal functor from the syntax category  to the category of monotone relations (Definition~\ref{def:undir-sem-int}). We will later prove that the inequational theory is \emph{sound and complete} for this interpretation. 
\begin{definition}
\label{def:undir-syn-cat}
We call $\catSyn$ the strict symmetric monoidal category freely obtained from a single generating object $1$ and the following generating morphisms, depicted as string diagrams \cite{selinger2010survey}:
\begin{equation}\label{eq:catsyndef-generators}
    \input{generate-ar.tikz}
\end{equation}
We write $\catSata$ for the category $\catSyn$ quotiented by
\footnote{Strictly speaking, since $\axBARel$ is an \emph{in}equational theory, the result is a 2-category where morphisms are posets (thus poset-enriched categories). For simplicity we choose to present it as a category, and reason about the ordering on morphisms externally. The same strategy is applied to the poset-enriched category $\catBProf$ with ordering $\subseteq$. }
the axioms of Figure~\ref{fig:ineq-axiom}. For a review of how to form this quotient and, more broadly, of the general methodology of presenting symmetric monoidal categories via (in)equational theories, we refer the reader to~\cite{bonchi2017functorial}, in particular Section 3 of that work.
\end{definition}
We shall identify the monoidal product $1^{\tensor n}$ of $n$ copies of $1$ with the natural number $n$, and write $\begin{tikzpicture}
	\begin{pgfonlayer}{nodelayer}
		\node [style=none] (0) at (-0.75, 0) {};
		\node [style=none] (1) at (0, 0) {};
		\node [style=none] (2) at (0.75, 0) {};
	\end{pgfonlayer}
	\begin{pgfonlayer}{edgelayer}
		\draw (0.center) to (1.center);
		\draw (1.center) to (2.center);
	\end{pgfonlayer}
\end{tikzpicture}
$ for the identity morphism on $1$ and $\swap$ for the symmetry morphism on $1$. We depict the monoidal product $c\tensor d$ of two diagrams as their vertical juxtaposition, and the composition $c\then d$ as the horizontal concatenation of $c$ and $d$. We now describe the interpretation of string diagrams as monotone relations.  
\begin{definition}
\label{def:undir-sem-int}
Let $\intBProf{\cdot} \colon \catSyn \to \catBProf$ to be the (lax) symmetric monoidal functor freely obtained by mapping the generating object $1$ to $\Bool$ and the generating morphisms of $\catSyn$ as follows:
\begin{align*}
    \intBProf{\input{copy.tikz}} & = \{ (x, (y_1, y_2)) \mid x \leq y_i, ~\text{and}~ x, y_1, y_2 \in \Bool \} 
    & \quad & 
    \intBProf{\begin{tikzpicture}
	\begin{pgfonlayer}{nodelayer}
		\node [style=none] (0) at (-0.5, 0) {};
		\node [style=black dot] (1) at (0.5, 0) {};
	\end{pgfonlayer}
	\begin{pgfonlayer}{edgelayer}
		\draw (0.center) to (1);
	\end{pgfonlayer}
\end{tikzpicture}
}  \ = \ \{ (x, \bullet) \mid x \in \Bool \}
    \\
    \intBProf{\input{cocopy.tikz}} & = \{ ((x_1, x_2), y) \mid x_i \leq y, ~\text{and}~ x_1, x_2, y \in \Bool \}      & \qquad &
    \intBProf{\begin{tikzpicture}
	\begin{pgfonlayer}{nodelayer}
		\node [style=none] (2) at (0.45, 0) {};
		\node [style=black dot] (4) at (-0.55, 0) {};
	\end{pgfonlayer}
	\begin{pgfonlayer}{edgelayer}
		\draw (4) to (2.center);
	\end{pgfonlayer}
\end{tikzpicture}
}  \ = \ \{ (\bullet, x) \mid x \in \Bool \} 
    \\
    \intBProf{\input{conj.tikz}} & = \{ ((x_1, x_2), y) \mid x_1 \land x_2 \leq y, ~\text{and}~ x_1, x_2, y \in \Bool \} 
    & \qquad &
    \intBProf{\begin{tikzpicture}
	\begin{pgfonlayer}{nodelayer}
		\node [style=none] (0) at (0.4, 0) {};
		\node [style=empty dot] (2) at (-0.6, 0) {};
	\end{pgfonlayer}
	\begin{pgfonlayer}{edgelayer}
		\draw (2) to (0.center);
	\end{pgfonlayer}
\end{tikzpicture}
} \ = \ \{ (\bullet, 1) \} 
    \\
    \intBProf{\input{coconj.tikz}} & = \{ (x, (y_1, y_2)) \mid x \leq y_1 \lor y_2, ~\text{and}~ x, y_1, y_2 \in \Bool \} 
    & \qquad &
\intBProf{\begin{tikzpicture}
	\begin{pgfonlayer}{nodelayer}
		\node [style=none] (0) at (-0.5, 0) {};
		\node [style=empty dot] (1) at (0.5, 0) {};
	\end{pgfonlayer}
	\begin{pgfonlayer}{edgelayer}
		\draw (0.center) to (1);
	\end{pgfonlayer}
\end{tikzpicture}
}  \ = \ \{ (0, \bullet) \} 
\end{align*}
\noindent Note that, as $\intBProf{\cdot}$ is freely generated, $\intBProf{\begin{tikzpicture}
	\begin{pgfonlayer}{nodelayer}
		\node [style=none] (0) at (-0.75, 0) {};
		\node [style=none] (1) at (0, 0) {};
		\node [style=none] (2) at (0.75, 0) {};
	\end{pgfonlayer}
	\begin{pgfonlayer}{edgelayer}
		\draw (0.center) to (1.center);
		\draw (1.center) to (2.center);
	\end{pgfonlayer}
\end{tikzpicture}
} \coloneqq \{ (x, y) \mid x \leq y ~\text{and}~ x, y \in \Bool \}$, $\intUndirProf{c \then d} \coloneqq \intUndirProf{c} \relThen \intUndirProf{d}$, and $\intUndirProf{c \tensor d} \coloneqq \intUndirProf{c} \times \intUndirProf{d}$. 
\end{definition}

In general, the interpretation provided by $\intBProf{\cdot}$ is not faithful: distinct morphism of $\catSyn$ may be mapped to the same monotone relation. However, when we quotient the syntax by the inequational theory $\axBARel$ (Figure~\ref{fig:ineq-axiom}) the mapping $\catSata \to \catBProf$ we obtain from the quotiented syntax to the semantics is still a symmetric monoidal functor and is now faithful -- this is what we mean by soundness (functoriality) and completeness (faithfulness) of $\axBARel$.
\begin{figure}[htb]
    \centering
    \begin{subfigure}{0.7\textwidth}
    \centering
       \input{axiom-black.tikz}
       \caption{Degenerate bimonoid}
       \label{fig:undir-black-ax}
    \end{subfigure}
    \hfill
    \begin{subfigure}{0.7\textwidth}
       \input{axiom-white.tikz}
       \caption{Frobenius algebra}
       \label{fig:undir-white-ax}
    \end{subfigure}
    \hfill
    \begin{subfigure}{1\textwidth}
    \centering
       \input{axiom-BW.tikz}
       \caption{White-black and black-white bimonoids}
       \label{fig:undir-BW-ax}
    \end{subfigure}
    \hfill
    \begin{subfigure}{1\textwidth}
    \centering
       \input{axiom-BW-adj.tikz}
       \caption{Adjunctions}
       \label{fig:undir-BW-adj-ax}
    \end{subfigure}
    \hfill
    %
    %
    \caption{The equational theory $\axBARel$}
    \label{fig:ineq-axiom}
\end{figure}
\begin{theorem}[Soundness and Completeness]\label{thm:sound-complete-monotone-rel-FBA}
For any $c, d \in \catSynDef$, $c \leq_{\catSata} d$ if and only if $\intBProf{c} \subseteq \intBProf{d}$. 
\end{theorem}

The proof will be given in Section~\ref{sec:sound-complete}. We now highlight some features of $\axBARel$. They will come handy when using it to study Boolean satisfiability and logic programming, in the next two sections.
\begin{itemize}
    \item (A1)-(A4) state that $\copyy$ and $\discard$ form a commutative comonoid, while (A5)-(A8) states that $\cocopy$ and $\codiscard$ form a commutative monoid. In standard algebraic syntax, we can implicitly associate multiple applications of a monoid operation say, to the left, in order to avoid a flurry of parentheses, \emph{e.g.} $a \lor b \lor c = (a \lor b) \lor c$. Thanks to laws (A1)-(A4), the same principle applies to an associative $2\to 1$ diagram of $\axBARel$. This justifies introducing generalised $n$-ary operations $\scalebox{0.8}{\input{n-ary-monoid.tikz}}$ as syntactic sugar. The same works for comonoids, by simply flipping the definitions as $\scalebox{0.8}{\input{n-ary-comonoid.tikz}}$.
    \item (B1)-(B10) state that $\conjj, \unit, \coconj, \counit$ forms a Frobenius algebra~\cite{carboni1987cartesian}.
    A first important consequence is that this structure makes 
$\catSata$ a \emph{compact closed category}~\cite{kelly1980coherence}, a category in which one can interpret fixed-points or iteration~\cite{hasegawa2012models}. 
%
Specifically, we have a cup $\input{cup.tikz}$ and cap $\input{cap.tikz}$ defined as $\input{cup-white-1.tikz}$ and $\input{cap-white-1.tikz}$ respectively, satisfying the `snake' equation $\scalebox{0.8}{\input{yanking-eq.tikz}}$. Note that axiom (A14) is defined with these as syntactic sugar.

Moreover, Frobenius algebras greatly reduce the mental load in keeping track of different ways of composing $\conjj, \unit, \coconj, \counit$, as they turn out to be equal. This simplification is more formally stated as the \emph{spider theorem} --- see e.g.~{\cite[Theorem 5.21]{heunen2019categories}}.
\begin{proposition}[Spider theorem]\label{thm:spider}
If $\conjj, \unit, \coconj, \counit$ form a commutative Frobenius structure, any \emph{connected} morphism $m \to n$ made out of $\conjj, \unit, \coconj, \counit$, identities, symmetries, using composition and the monoidal product, is equal to the normal form \input{spider-normal-form.tikz}. 
\end{proposition}
In the statement, the word `connected' refers to the diagram as a undirected graph: we assume there is at least one path made up of wires, between any two white nodes in the diagram.

In our axiomatic theory, we can simplify the normal form of Theorem~\ref{thm:spider} even further to get rid of loops. Making crucial use of axiom (A14), we can derive \inlineequation[eq:not-special]{\input{not-special.tikz}}. 
Note it is unusual for a Frobenius algebra to satisfy \eqref{eq:not-special} as those that have appeared in the literature satisfy a version of \eqref{eq:not-special} with the rhs set to the identity (\emph{e.g.}~\cite{zanasi2015interacting}) or to $\counit\;\unit$ (\emph{e.g.}~\cite{coecke2010compositional}). 
Thanks to the spider theorem and \eqref{eq:not-special}, we can see any \emph{simply connected} composition of white monoids and comonoids as a single node, whose only relevant structure is its number of dangling wires on the left and on the right. The normal form guarantees that we can unambiguously depict any such morphism $m\to n$ as a \emph{spider} with $m$ wires on the left and $n$ on the right: $\input{spider.tikz}$.
With this syntactic sugar, the axioms of the Frobenius structure $\conjj, \unit, \coconj, \counit$ can be concisely expressed via the following \emph{spider fusion} scheme:
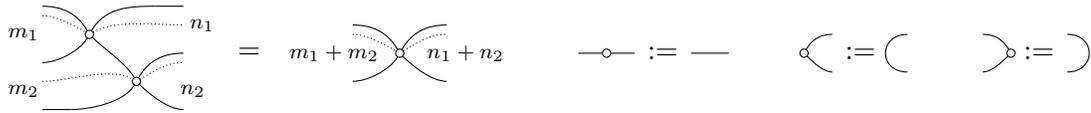
\[\input{spider-fusion.tikz}\]

    \item The C block of axioms state that $(\copyy, \discard, \unit, \conjj)$ and $(\coconj, \counit, \codiscard, \cocopy)$ form two bimonoids. Semantically, these hold when the underlying poset is a lattice (has binary meets and joins). They also imply that $\catSata$ contains as a monoidal subcategory, the category of Boolean matrices with the direct sum as monoidal product (\emph{cf.}~Appendix~\ref{sec:matrix-diagrams}).

    \item The D axioms present a number of adjunctions in the 2-categorical sense (or rather, Galois connections, because the 2-morphisms are simply inclusions). Two morphisms $f\colon X \to Y$ and $g\colon Y\to X$ are adjoint if $id_X \leq f\then g$ and $g\then f\leq id_Y$. We write this as $f\dashv g$ ($f$ left adjoint to $g$). The situation can be summarised by the following six adjunctions:
    \begin{align*}
        \conjj \dashv \copyy  \dashv \cocopy \dashv \coconj \qquad \qquad \qquad 
        \unit \dashv \discard \dashv \codiscard \dashv \counit
    \end{align*}
    The adjunctions $\copyy \dashv \cocopy$ and $\discard \dashv \codiscard$ give the defining inequations of cartesian bicategories~\cite{carboni1987cartesian}.
    
    To understand where the other adjunctions come from, it is helpful to adopt a semantic perspective. For example, recall that  $\intBProf{\conjj} = \{((x_1,x_2),y) \mid x_1\land x_2 \leq y\}$ and $\intBProf{\copyy} = \{(x,(y_1,y_2)) \mid x\leq y_1, x\leq y_2\} = \{(x,(y_1,y_2)) \mid x\leq y_1 \land y_2\}$. Thus, one can see the adjunction $\conjj \dashv \copyy$ as arising from the duality between the two different ways of turning the monotone function $\land : \Bool^2\to \Bool$ into a monotone relation (by placing it on either side of the $\leq$ symbol in the corresponding set comprehension).
    
    Semantically, they guarantee that the underlying poset $\Bool$ is a \emph{lattice}, \emph{i.e.}, has binary meets and joins. Together with the Frobenius axioms (B9)-(B10), they imply that $\Bool$ is a complemented distributive lattice, in other words, a Boolean algebra.
\end{itemize}
\begin{remark}
Note that we did not aim for a minimal presentation of the theory. There are obvious redundancies: for example, axiom (C4) immediately implies (D4). Similarly, many equations of the B block can be derived from those of the D block and the corresponding properties of the black generators in the A block. We have included these redundancies in order to avoid burdening the reader with too many lemmas deriving equations that we will need often, and to highlight key algebraic structures that occur in related theories (\emph{e.g.} bimonoids). 
\end{remark}




\section{$\catSata$ as the Algebra of Satisfiability}\label{sec:sat}

\subsection{From diagrams to CNF formulas}

The language we have introduced in the previous section can be seen as a diagrammatic notation to reason about Boolean formulas expressed in conjunctive normal form. As hinted at in the introduction, we can associate to every diagram a CNF formula, with the relational semantics of the diagram given by its set of satisfying assignments. The generators can be interpreted as follows:
\begin{align*}
\copyy &\mapsto (\lnot x \lor y_1) \land (\lnot x \lor y_2)  \qquad &\discard\mapsto 1 \qquad\cocopy &\mapsto (\lnot x_1 \lor y) \land (\lnot x_2 \lor y)\qquad &\codiscard \mapsto 1 \\
 \coconj &\mapsto \lnot x \lor y_1 \lor y_2\qquad &\unit\mapsto y\qquad  \conjj &\mapsto \lnot x_1\lor \lnot x_2\lor y\qquad &\counit\mapsto \lnot x
\end{align*}
Note that an identity wire does not identify two variables $x=y$; rather, it represents the clause $\lnot x\lor y$. 
The monoidal product is conjunction: if diagrams $d_1$ and $d_2$ correspond to some formulas $f_1$ and $f_2$ respectively, then their monoidal product corresponds to $f_1\land f_2$. 
Composition is more subtle. As a plain identity wire represents the clause $\lnot x\lor y$, attaching two wires does not identity two variables, but merely adds this clause and \emph{existentially quantifies} over $x$ and $y$. More precisely, given a formula $f$, let $\exists x f := f[x=1]\lor f[x=0]$, where $f[x=b]$ denotes $f$ in which we have replaced all occurrences of $x$ by the Boolean value $b$. Further assume, for simplicity, that we compose $\Diagbox{c}{m}{} \mapsto e$ and $ \Diagbox{d}{}{n}\mapsto f$ along a single wire: let $y$ be the variable that denotes the single wire in the codomain of $c$ and $x$ the single wire in the domain of $d$; then   
$\input{cnf-composition.tikz}$.

One of the key observations is that the generalised white nodes introduced above, correspond to an arbitrary clause with $m$ negative literals and $n$ positive literals: $\input{spider.tikz} \quad \mapsto\quad  \lnot x_1\lor \cdots \lor \lnot x_m\lor y_1 \lor \cdots \lor y_n$.
Note that if we disallow the use of $\coconj$, we can only represent clauses with at most one positive literal, better known as \emph{Horn clauses}. We will make extensive use of Horn clauses in the next section, to give a diagrammatic account of propositional logic programs. 
\begin{remark}
We could have given an infinitary presentation of the same symmetric monoidal category with one generator for each clause containing $m$ negative literals and $n$ positive ones (keeping the same black generators).
We were able to give a presentation with a finite number of generators (and inequations) - with at most three wires - for the same reason that SAT can be reduced to $3$-SAT. 
\end{remark}

\subsection{Picturing SAT}\label{sec:sat-encoding}
\begin{wrapfigure}{r}{0.3\textwidth}
   \scalebox{0.8}{\input{sat-encoding-intuition.tikz}}
\end{wrapfigure}
The payoff is that we can use the diagrammatic calculus to encode the satisfiability of arbitrary  Boolean formulas in a very natural way. Given a Boolean formula $f$ in CNF,  SAT($f$) can be depicted as a closed diagram (one of type $0\to 0$) of the form on the right, where the white nodes in the middle encode clauses and the black nodes take care of the variable management (copying and deleting) needed to connect variables to the clauses in which they appear, \emph{i.e.} to the corresponding white nodes. The boxes are just permutations of the wires.

There are - up to equality - only two diagrams of type $0\to 0$, namely $\input{id-0.tikz}$ (the empty diagram) and $\begin{tikzpicture}
	\begin{pgfonlayer}{nodelayer}
		\node [style=empty dot] (41) at (0.25, 0) {};
		\node [style=empty dot] (42) at (-0.25, 0) {};
	\end{pgfonlayer}
	\begin{pgfonlayer}{edgelayer}
		\draw (42) to (41);
	\end{pgfonlayer}
\end{tikzpicture}
$, interpreted as the relations $\{\bullet\}$ and $\empset$, respectively.
Following the completeness of our equational theory, the soundness of our encoding of SAT implies that a formula $f$ is satisfiable iff the associated diagram is equal to $\input{id-0.tikz}$, as stated in Theorem~\ref{thm:diagrammatic-sat} below.
Furthermore, using our equational theory, we can derive the (un)satisfiability of a SAT instance purely algebraically, justifying the claim that our diagrammatic calculus captures the algebra of satisfiability. 

Let us describe the encoding of SAT($f$) in more detail. Another way to understand the diagram above is as two relations that represent the negative and positive occurrences of variables - given by the black nodes on the left and on the right respectively - in each clauses - given by the white nodes at the center. We can encode these two incidence relations as Boolean matrices: given a CNF formula $f$ with $n$ free variables and $k$ clauses (both arbitrarily ordered), let $N(f)$ be the $k\times n$ Boolean matrix with $N(f)_{i,j} = 1$ iff $\lnot x_i$ appears in the $j$-th clauses of $f$, and $P(f)$ be the $n \times k$ Boolean matrix where $P(f)_{j,i} = 1$ iff $x_i$ appears in the $j$-th clauses of $f$. 
In other words, $N(f)$ encodes negative occurrences of variables, and $P(f)$ their positive occurrences in each clause. 

We use the standard diagrammatic representation of matrices using a commutative bimonoid to define diagrams for $N(f)$ and $P(f)$ ---a translation which is recalled in Appendix~\ref{sec:matrix-diagrams}. 
We use below the same names for the associated diagrams: $N(f)$ denotes the corresponding matrix encoded using $\copyy, \discard, \unit, \conjj$, and $P(f)^T$ the transpose of $P(f)$ encoded using $\coconj, \counit, \codiscard, \cocopy$. 

\begin{theorem}\label{thm:diagrammatic-sat}
A CNF formula $f$ is satisfiable if and only if 
   $\ \scalebox{0.7}{\input{sat-encoding.tikz}} \; = \; \input{id-0.tikz}$. 
\end{theorem}
\begin{example}\label{ex:sat-2}
Consider $f = (\lnot x\lor y)\land (x\lor y)\land (x\lor \lnot y)\land (\lnot x\lor \lnot y)$. 
Following the notation in Theorem~\ref{thm:diagrammatic-sat}, let $N(f) = \input{cnf-eg-neg.tikz}$ where the left wires represent the \emph{negative} occurrences of $x$ and $y$ respectively, from top to bottom, and the right wires represent the four clauses $(\lnot x\lor y), (x\lor y), (x\lor \lnot y)$, and  $(\lnot x\lor \lnot y)$ also from top to bottom. Similarly, let $P(f)^T = \input{cnf-eg-pos.tikz}$ encode all positive occurrences of $x$ and $y$ (right wires) into each of the clauses (left wires). Therefore SAT($f$) can be encoded diagrammatically as
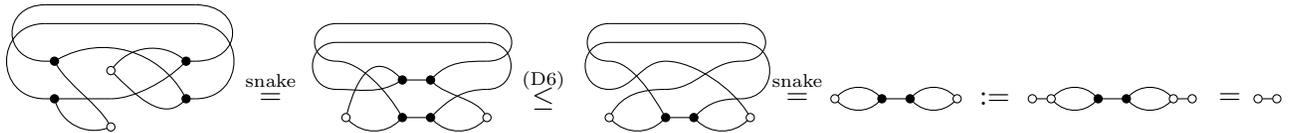
\[\input{sat-instance-ex-2.tikz}\]
As we will explain below, in Subsection~\ref{sec:sat-solving}, the central inequality can be seen as an instance of \emph{resolution}.
\end{example}
\begin{example}\label{ex:sat-1}
$(\lnot x\lor y)\land (x\lor \lnot y)$. Here, we have two variables that appear precisely once as negative and positive literals in two clauses:
$P(f)$ and $N(f)$ are both equal to a single identity wire.
This gives the following, which we can prove is satisfiable equationally:
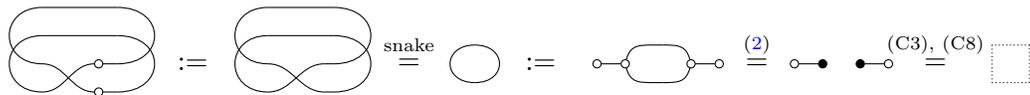
\[\input{sat-instance-ex-1.tikz}\]
\end{example}
\begin{remark}
Inequalities can be used to prove that a formula is (un)satisfiable more efficiently. As we have explained above, there are - up to equality - only two diagrams of type $0 \to 0$, namely  $\begin{tikzpicture}
	\begin{pgfonlayer}{nodelayer}
		\node [style=empty dot] (41) at (0.25, 0) {};
		\node [style=empty dot] (42) at (-0.25, 0) {};
	\end{pgfonlayer}
	\begin{pgfonlayer}{edgelayer}
		\draw (42) to (41);
	\end{pgfonlayer}
\end{tikzpicture}
 \leq \input{id-0.tikz}$. Thus, if $d(f) \leq \begin{tikzpicture}
	\begin{pgfonlayer}{nodelayer}
		\node [style=empty dot] (41) at (0.25, 0) {};
		\node [style=empty dot] (42) at (-0.25, 0) {};
	\end{pgfonlayer}
	\begin{pgfonlayer}{edgelayer}
		\draw (42) to (41);
	\end{pgfonlayer}
\end{tikzpicture}
$ for some formula $f$, we can strengthen this inequality to an equality and conclude that $f$ is unsatisfiable. Similarly, if we can show that $\input{id-0.tikz} \leq d(f)$ then $d(f)=\input{id-0.tikz}$ so $f$ is satisfiable.
\end{remark}
\begin{remark}
\label{rmk:monotone-sat}
It is helpful at this point to pause the development in order to examine the relationship between monotonicity and satisfiability. As we have already pointed out, at first glance, it appears that restricting ourselves to monotone relations is insufficient to study satisfiability. Indeed, monotone formulas (\emph{i.e.} those which do not use negation) are always satisfiable. Correspondingly, at the level of the diagrammatic syntax, we do not have access to a primitive negation operation, which non-monotonic. As a result, we can never explicitly enforce that a variable is the negation of another. So how can we encode SAT without negation? The key idea is that we do have a restricted form of negation, which corresponds to changing the direction of wires, using the cups and caps.

Even though the diagrammatic translation is straightforward, there first appears to be a mismatch between a given CNF formula and the semantics of the diagram representing a SAT instance. A variable appearing as a positive and negative literal in a given formula $f$, is depicted as two wires -- one for its positive occurrences and one for its negative occurrences.
When we wish to existentially quantify over a variable of a SAT instance, we simply trace out the corresponding two wires, connecting them.
In the relational semantics, the resulting loop joining the two wires represents 
the additional clause $\lnot y\lor x$, where $y$ always appear positively in other clauses, and $x$ always appears negatively. This leaves the possibility of setting both variables to $0$, which satisfies $\lnot y\lor x$ but does not match the intuition that these two variables should behave like negated version of each other. However, the assignment $y=0, x=0$ does not affect the satisfiability of the overall diagram representing the corresponding SAT instance\footnote{In light of Theorem~\ref{thm:diagrammatic-sat}, we say that a closed diagram $d$ is satisfiable when $\intBProf{d} = \{\bullet\}$.}. Indeed, an issue can only arise if the assignment $y=0, x=0$ is required to make the diagram satisfiable. But this can never be the case because $y$ always appears positively in any clause other than $\lnot y\lor x$. This implies that the satisfiability of the clauses in which $y$ appears will only depend on the other variables and therefore leave the overall satisfiability of the diagram unaffected.
%
\end{remark}

\subsection{SAT-solving, diagrammatically?}\label{sec:sat-solving}

Now that we know how to represent SAT instances in our diagrammatic language, it is useful to give meaning to the axioms in the light of this interpretation. First, the (co)associativity, (co)unitality and (co)commutativity of the various (co)monoids are equivalent to the associativity, unitality and commutativity of $\land$ and $\lor$ in propositional logic. The loop removal axiom from Fig.~\ref{fig:undir-black-ax} admits a natural interpretation in terms of satisfiability: if $\exists x (C\lor x\lor \lnot x)$ for some clause $C$, we can always satisfy $x\lor \lnot x$ and therefore we can remove constraint imposed by the whole clause $C\lor x\lor \lnot x$, regardless of the assignment of values to the other variables appearing in $C$. This axiom allows us to remove this kind of redundant constraints. 

That the white nodes form a Frobenius algebra is more interesting: it stems from the fact that $\exists x.(C\lor x)\land (D\lor \lnot x)$ allows us to deduce the clause $C\lor D$, and is thus the diagrammatic translation of the \emph{resolution} rule~\cite{robinson1965machine} in classical propositional logic! 

Resolution has come to occupy a central role in SAT-solving as modern SAT-solvers can be thought of as resolution proof systems. It is well-known that solvers based on a variant of the Davis-Putnam-Logemann-Loveland (DPLL) algorithm~\cite{davis1960computing,davis1962machine} implement a form of \emph{tree-like} resolution, while those that are augmented with clause learning (\emph{e.g.} CDCL) correspond to more general forms of resolution proofs, that allow for the reuse of learned clauses~\cite{pipatsrisawat2011power}. 
Tree-like resolution proofs can be easily interpreted diagrammatically. Recall from our interpretation of SAT instances that clauses are simply white nodes whose left (\emph{resp.} right) wires encode negative (\emph{resp.} positive) literals. Thus, if we have two clauses $C\lor x$ and $\lnot x \lor D$, the corresponding diagram will contain a subdiagram equal to the leftmost one in~\eqref{eq:satexample} below. The black nodes in the subdiagram represent the fact that $x$ might also appear (negatively or positively) in other clauses. 
From this, resolution allows us to deduce $C\lor D$; tree-like resolution proofs further impose that, if we apply a deduction step, we do not use the original clauses $C\lor x$ and $\lnot x \lor D$ again. This has a simple diagrammatic counterpart---we can apply axiom (D6) to connect the two white nodes corresponding to the two clauses in premise of a resolution step directly, effectively replacing them with a single node that represents the learned clause:
\begin{equation}\label{eq:satexample}
\scalebox{0.8}{\input{tree-resolution.tikz}}
\end{equation}
The remaining wire for other uses of $x$ is unaffected. We have already used the same idea in Example~\ref{ex:sat-2} above, in order to resolve the subdiagrams corresponding to the clauses $\lnot x\lor y$ and $x\lor y$.

General resolution -- where we are allowed to keep the two clauses in the premise of the rule -- also has a diagrammatic counterpart. In fact, it is implemented as a sequence of equalities (instead of the inequality we have used above to interpret tree-like resolution), reflecting the fact that the process does not forget any information. The idea is to first copy the two clauses $C\lor x$ and $\lnot x \lor D$ to leave them available for any further resolution with other clauses that use the variable $x$. We illustrate this process with two clauses (containing only three literals for simplicity this time) below:
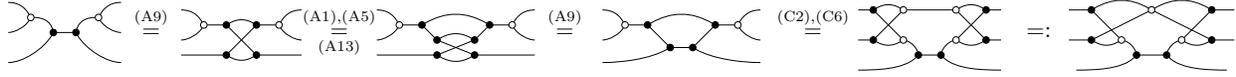
\[\scalebox{0.8}{\input{general-resolution.tikz}}\]
In the resulting diagram, the quaternary white node represents the clause $C\lor D$, while the other two represent the original clauses $C$ and $D$. 


\section{$\catSata$ as a proof system for logic programming}
\label{sec:LP}
In this section we turn to logic programming. 
$\catSata$ allows us to compute the usual semantics of propositional definite logic programs via diagrammatic equational reasoning (Theorem~\ref{thm:LP-calculate}). 
More importantly, we obtain a decidable, purely equational procedure to decide the equivalence of program, by proving the equivalence of their representing diagrams in $\catSata$ (Theorem~\ref{thm:LP-eq}). 

First we briefly recall the basics of definite logic programming, and refer the readers to \cite{lloyd2012foundations} for details. 
Logic programming is a programming paradigm which reduces computation to proof search in formal logic. It is widely used in knowledge representation \cite{baral1994logic}, database theory \cite{ceri1989you}, constraint solving \cite{jaffar1987constraint}. We focus on the simplest yet prototype class of logic programs, called \emph{propositional definite logic programs}, formally defined as follows. We fix some finite set of atoms $\At = \{ a_1, \dots, a_k \}$. A \emph{Horn clause} $\clauseD$ is a formula of the form $b_1, \dots, b_k \to a$, where $a, b_1, \dots, b_k \in \At$ and $b_1 \dots, b_k$ are distinct. Here the atom $a$ and the set $\{ b_1, \dots, b_k \}$ are called the \emph{head} ($\head(\clauseD)$) and the \emph{body}  ($\body(\clauseD)$) of the clause, respectively. 
If a clause has an empty body, then it is also called a \emph{fact}. 
A \emph{propositional definite logic program} (or simply \emph{logic program}) is a finite set of Horn clauses based on $\At$. 
Intuitively, a clause $b_1, \dots, b_k \to a$ reads `if $b_1, \dots, b_k$ hold, then $a$ also holds', which gives the intuitive meaning of a program $\progL$ to be the set of all atoms that are deducible from $\progL$ in finitely many steps. 
This is formally defined as the \emph{least Herbrand
model semantics}. 
An \emph{interpretation} $\intI$ over $\At$ is an assignment $\intI \colon \At \to \Bool$. 
For simplicity, we will also write an interpretation $\intI$ over $\At$ as a subset $\{ a \in \At \mid \intI(a) = 1 \}$.

A logic program $\progL$ defines an \emph{immediate consequence operator} $\imCnsqOp{\progL} \colon \Bool^{\At} \to \Bool^{\At}$ such that, given an interpretation $\intI$,
\[
\imCnsqOp{\progL}(\intI)(a) =
\begin{cases}
1 & \text{there exists $\progL$-clause}~ b_1, \dots, b_k \to a ~\text{such that}~ \intI(b_1) = \cdots = \intI(b_k) = 1 \\
0 & \text{otherwise}
\end{cases}
\]
The \emph{Herbrand model} (denoted as $\leastHMod{\progL}$) is then the least fixed point of $\imCnsqOp{\progL}$. Such least fixed point always exists because $\Bool^\At$ carries a complete lattice structure, and $\imCnsqOp{\progL}$ is monotone on it. 
%
We use the following notation: $\land$ and $\lor$ are defined pointwise on $\At$, $\funcOne$ and $\funcZero$ are the constant functions mapping arbitrary $a \in \At$ to $1$ and $0$, respectively.

From a deductive perspective, a logic program $\progL$ can be viewed as a black box: it takes a set of atoms $A$ as input, and outputs the set of atoms that are deducible from $A$ using $\progL$. This is formalised as the \emph{consequence operator} $\cnsqOp{\progL} \colon \Bool^\At \to \Bool^\At$ of program $\progL$, such that
$\cnsqOp{\progL}(A)$ is the least Herbrand model of the program $\progL \cup \{ \to a \mid a \in A \}$. In particular, $\cnsqOp{\progL}(\empset)$ is exactly $\leastHMod{\progL}$. 
\begin{example}
\label{eg:LP-cnsq-op}
Consider the program $\progP = \{ \to a; b \to d; c \to d; c,d \to b \}$ over $\At = \{ a, b, c, d \}$. 
Its least Herbrand model $\leastHMod{\progP}$ is $\{ a \}$. %
The immediate consequence operator $\imCnsqOp{\progP}$ maps, for instance, $\{ c \}$ to $\{ a, d \}$. 
The consequence operator $\cnsqOp{\progP}$ maps, for instance, $\{ c \}$ to $\{ a, b, c, d \}$. 
\end{example}
%

We will represent (immediate) consequence operators of logic programs via our diagrammatic language.
Immediate consequence operators have been studied from a string diagram perspective in \cite{gu2021functorial}. Compared with \eqref{eq:im-conseq-diag} below, there each clause is represented by a box, while our diagrammatic language enables us to `open the boxes' by representing a definite clause as $\input{conj-k.tikz}$.  
The key observation is the intuitive reading of clause $b_1, \dots, b_k \to a$ as $b_1\land \dots \land b_k \leq a$, thus representable as the $\catSyn$-morphism $\input{k-conj.tikz}$.
\begin{definition}
\label{def:imcnsq-diagram}
The \emph{immediate consequence diagram} $\imCnsqDiag{\progL}$ for program $\progL$ is an $n \to n$ morphism defined as follows, where the $i$-th input and output wires stands for atom $a_i$, for $i = 1,\dots, n$.
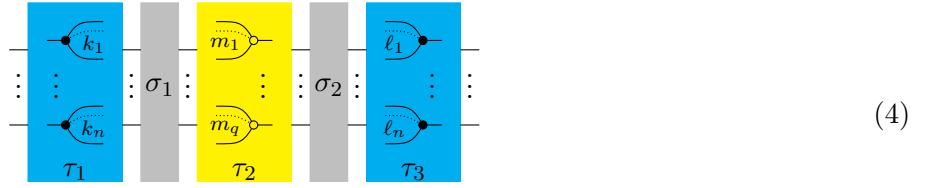
\begin{equation}\label{eq:im-conseq-diag}
    \input{im-conseq.tikz}
\end{equation}
The two blue blocks $\tau_1$ and $\tau_3$ consist respectively of copiers $\input{copy.tikz}$ and cocopiers $\input{cocopy.tikz}$: $k_i$ and $\ell_i$ are the numbers of appearance of atom $a_i$ in all the \emph{bodies} and heads of $\progL$-clauses, respectively. 
The yellow block $\tau_2$ is the parallel composition of $\conjj$, where $q$ is the number of clauses in $\progL$, and for each clause $\clauseD_j$, $m_j$ is the size of $\body(\clauseD_j)$. 
The grey blocks $\sigma_1$ and $\sigma_2$ consist of appropriately many $\swap$ that change the wire orders to match the domains and codomains of the blue and yellow blocks. In particular, the $i$-th input (\textit{resp.} output) wire is connected to the $j$-th $\conjj$ if $a_i \in \body(\clauseD_j)$ (\textit{resp.} $a_i = \head(\clauseD)$).
\end{definition}

\begin{wrapfigure}{r}{0pt}
    \input{conseq-diag.tikz}
\end{wrapfigure}
Based on this, we can compositionally represent the consequence operator $\imCnsqOp{\progL}$.
\begin{definition}
\label{def:cnsq-diagram}
The \emph{consequence diagram} for program $\progL$ is a $\catSyn$-morphism $\cnsqDiag{\progL} \colon n \to n$ defined on the right, where $\imCnsqDiag{\progL}$ is the immediate consequence diagram from Definition~\ref{def:imcnsq-diagram}.
\end{definition}
Note that the interpretation $\intBProf{\imCnsqDiag{\progL}}$ is not exactly the operator $\imCnsqOp{\progL}$: the former is a monotone relation while the latter is a monotone function. Nevertheless we say a monotone relation $R \subseteq X \times Y$ \emph{represents} a monotone function $f \colon X \to Y$ if for arbitrary $x \in X$ and $y \in Y$, $(x, y) \in R$ if and only if $f(x) \leq y$. 
Crucially, if monotone relations $R$ and $S$ represent monotone functions $f$ and $g$ respectively, then $f = g$ if and only if $R = S$. 
%
\begin{proposition}\label{prop:cnsq-represent}
The monotone relation $\intBProf{\imCnsqDiag{\progL}}$ (resp. $\intBProf{\cnsqDiag{\progL}}$) represents the operator $\imCnsqOp{\progL}$ (resp. $\cnsqOp{\progL}$).
\end{proposition}

\begin{example}
Recall $\progP$ from Example~\ref{eg:LP-cnsq-op}. 
$\imCnsqDiag{\progP}$ is given below left, coloured as in Definition~\ref{def:imcnsq-diagram}.
\end{example}

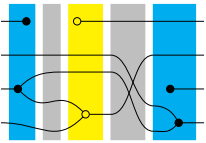
\begin{wrapfigure}{l}{0pt}
    \scalebox{0.9}{\input{eg-imcnsq-diag-1.tikz}}
\end{wrapfigure}
Now we are ready to present the two main results for logic programs. 
First, $\axBARel$ provides a diagrammatic calculus to compute $\cnsqOp{\progL}$, in particular to calculate the least Herbrand model as $\cnsqOp{\empset}$. One can represent an interpretation $\intI \in \setTwo^{\At}$ as a $\catSyn$-morphism $d^{\intI} \coloneqq d^{\intI}_1 \tensor \cdots \tensor d^{\intI}_n$ of type $0 \to n$, where each $d^\intI_i = \codiscard$ if $\intI(i) = 0$, and $d^\intI_i = \unit$ if $\intI(i) = 1$.

\begin{theorem}
\label{thm:LP-calculate}
$\cnsqOp{\progL}(\intI)= \intJ$ if and only if $d
^\intI \then \cnsqDiag{\progL} = d^\intJ$.
\end{theorem}

Second, we can turn the problem of the equivalence of two logic programs into proving whether their consequence operator diagrams are equivalent. By Proposition~\ref{prop:cnsq-represent} and Theorem~\ref{thm:sound-complete-monotone-rel-FBA}, we have:
\begin{theorem}
\label{thm:LP-eq}
Given two logic programs $\progL_1$ and $\progL_2$ on $\At$, $\cnsqOp{\progL_1} = \cnsqOp{\progL_2}$ if and only if $\cnsqDiag{\progL_1} = \cnsqDiag{\progL_2}$.
\end{theorem}
It is worth observing that, as we will see in Section~\ref{sec:sound-complete}, the completeness proof is essentially an algorithm that rewrites an arbitrary $\catSyn$-morphism into a normal form diagram (Definition~\ref{def:normal-form}). This means that the procedure described in Theorem~\ref{thm:LP-eq} to establish equivalence of logic programs in terms of consequence operators is in fact decidable.

\begin{example}\label{ex:lpequiv}
Recall $\progP$ from Example~\ref{eg:LP-cnsq-op}, and let $\progQ = \{ \to a; a, b \to d; c \to b. \}$. We claim that $\progP$ and $\progQ$ are equivalent with respect to consequence operator. 
One can show this by proving the equivalence of their consequence diagrams, namely $\cnsqDiag{\progP} = \cnsqDiag{\progQ}$:
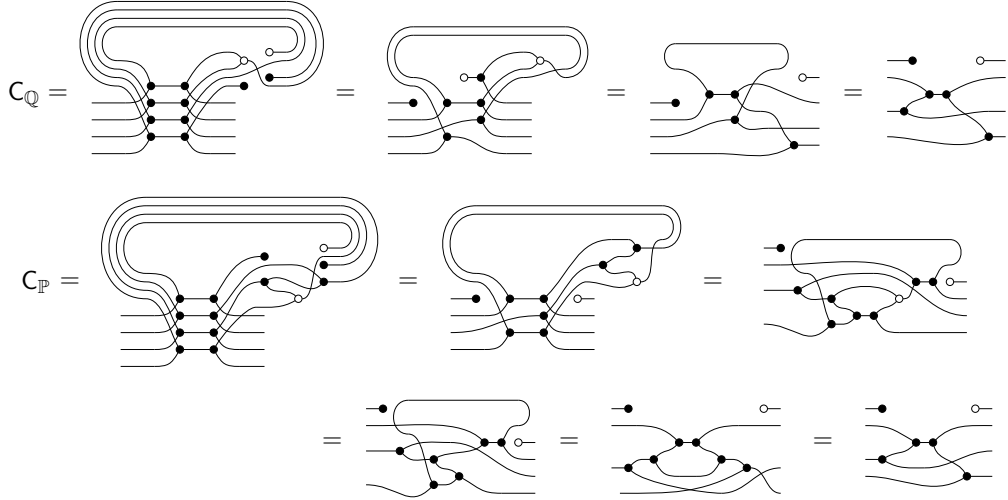
\[
\scalebox{0.9}{\input{eg-cnsq-eq-1.tikz}}
\]
\[
\scalebox{0.9}{\input{eg-cnsq-eq-2.tikz}}
\]
\end{example}


\begin{remark}
It is helpful to point out that if we interpret $\catSyn$ in the category $\catRel$ of relations (via a functor $[-] \colon \catSyn \to \catRel$), then $\cnsqDiag{\progL}$ --- an intuitive candidate to represent consequence operators --- is not interpreted as $\{(\intI, \intJ) \mid \cnsqOp{\progL}(\intI) = \intJ \}$. In short, for the semantics to be correct, one has to define $[\copyy] = \{ (x, (x, x)) \mid x \in \Bool \}$ and $[\cocopy] = \{ ((x_1, x_2), y) \mid x_1, x_2, y \in \Bool, x_1 \lor x_2 = y \}$. But then the program $\progL = \{ a \to a \}$ has $\cnsqDiag{\progL} = \scalebox{0.7}{\input{vacuous-feedback.tikz}}$, and $[\cnsqDiag{\progL}] = \{ (x, y) \mid x, y \in \Bool, x \leq y \} \neq \{ (x, x) \mid x \in \Bool \} =  \cnsqOp{\progL}$.
\end{remark}

\section{Soundness and Completeness}
\label{sec:sound-complete}

This section is devoted to the proof of Theorem~\ref{thm:sound-complete-monotone-rel-FBA}, i.e. that the equational theory $\axBARel$ is a complete axiomatisation of monotone relations over (finite) Boolean algebras. 

The `only if' direction (soundness) of Theorem~\ref{thm:sound-complete-monotone-rel-FBA} is straightforward: we just need to show that each axiom is a valid semantic (in)equality. We take (C2) as an example, but will omit the verification of the other axioms:
\begin{align*}
    \intBProf{\scalebox{0.8}{\input{eg-soundness-2.tikz}}} & = \{ ((x_1, x_2), (y_1, y_2)) \in \Bool^{4} \mid \exists z_1, z_2, z_3, z_4 \in \Bool: x_1 \leq z_1 \land z_2, x_2 \leq z_3 \land z_4, z_1 \land z_3 \leq y_1, z_2 \land z_4 \leq y_2 \} \\
    & = \{ ((x_1, x_2), (y_1, y_2)) \in \Bool^{4} \mid \exists z_3, z_4 \in \Bool: x_2 \leq z_3 \land z_4, x_1 \land z_3 \leq y_1, x_1 \land z_4 \leq y_2 \} \\
    & = \{ ((x_1, x_2), (y_1, y_2)) \in \Bool^{4} \mid \exists z_3, z_4 \in \Bool: x_1 \land x_2 \leq y_1, x_1 \land x_2 \leq y_2 \} \\
    & = \{ ((x_1, x_2), (y_1, y_2)) \in \Bool^{4} \mid x_1 \land x_2 \leq y_1 \land y_2 \} = \intBProf{\scalebox{0.8}{\input{eg-soundness-1.tikz}}}
\end{align*}

We now focus on the `if' direction (completeness). Before the technical developments, we provide a roadmap of the proof. We first show that we can make two simplifying assumptions: (1) the completeness statement about inequalities/inclusions can be reduced to one that involves only equalities (Lemma~\ref{lem:syn-plus-leq}); (2) We simplify the problem further by showing that we can restrict the proof of completeness to diagrams $m\to 0$ without loss of generality (Lemma~\ref{lem:bend-input}). 

To prove completeness for $m\to 0$ diagrams, we adopt a normal form argument that is common for the completeness proof of many diagrammatic calculi~\cite{piedeleu2018picturing,piedeleu2020automata}. 
In a nutshell, a normal form defines a certain syntactic representative for each semantic object in the image of the interpretation, and shows that if two diagrams have the same semantics, then they are both equal in $\axBARel$ to this syntactic representative. 
The structure of the normal form argument is as follows. 
\begin{itemize}
    %
    %
    %
    \item We give a procedure to rewrite any diagram into one in \emph{pre-normal form}, using equations of $\axBARel$ (Lemma~\ref{lem:pre-normal-form}). Intuitively, pre-normal form diagrams represent sets of clauses of the form $\lnot x_1\lor \cdots \lor \lnot x_n$.
    \item Each downward closed subset is the set of satisfying assignment for a minimal set of such clauses (Lemma~\ref{lem:min-system-eq}). The diagrammatic counterpart of this minimal set of clauses is our chosen normal form for each equivalence class of diagrams (Definition~\ref{def:normal-form}). We show that every diagram in pre-normal form is equal to one in normal form (Lemma~\ref{lem:normal-form}).
    \item Finally, since every semantic object has a canonical diagrammatic representative -- a normal form diagram -- completeness follows from the fact that two diagrams having the same semantics have the same normal form (Proposition~\ref{prop:one-one-clause-NF-dcset}).
    \end{itemize}

\noindent As explained above, completeness for equalities is enough to derive that for inequalities.  
This follows from the next lemma, which essentially says that in lattices, inequalities can be defined using equalities and meets. 
\begin{lemma}
\label{lem:syn-plus-leq}
For diagrams $c, d$, 
(i) if $\scalebox{0.8}{\input{plus-c-d.tikz}} = c$, then $c \leq d$ and (ii) if $\intBProf{c}\subseteq \intBProf{d}$ then  $\intBProf{\scalebox{0.8}{\input{plus-c-d.tikz}}}=\intBProf{c}$.
\end{lemma}
In addition, we can restrict the completeness proof to diagrams with codomain $0$. This is the consequence of axioms (B9)-(B10) which endow $\{ \coconj, \counit, \conjj, \unit \}$ with the structure of a Frobenius algebra and $\catSata$ with the structure of a compact closed category, allowing us to move wires from the left to the right and vice-versa.
\begin{lemma}
\label{lem:bend-input} 
For arbitrary natural numbers $m, n$, $\catSyn[m, n] \iso \catSyn[m+n, 0]$.
\end{lemma}

\subsection{Pre-normal form}
\label{subsec:pre-normal-form}
\begin{wrapfigure}{r}{0pt}
    \input{downward-normal-form.tikz}
\end{wrapfigure}
In this subsection, we show that every diagram is equivalent to one of \emph{pre-normal form}, that is, which factorises as on the right. We then explore their close connection with certain Boolean clauses, a key tool in completeness proof. 
Let us explain how to interpret the factorisation: each grey block represents a diagram which is formed only using the specified subset of generators (using both composition and monoidal product), possibly including some permutations of the wires. 
The equational theory allows us to rewrite any diagram to pre-normal form. 
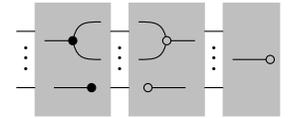
\begin{lemma}[Pre-normal form]
\label{lem:pre-normal-form}
Every diagram $m \to 0$ is equal to one of pre-normal form.
\end{lemma}

\begin{wrapfigure}{r}{0pt}
    \input{matrix-diagram.tikz}
\end{wrapfigure}
The first two blocks of pre-normal forms deserve special attention. We say that $d$ is a \emph{matrix diagram} if it factorises as on the right.
\noindent In our semantics, matrix diagrams $m\to n$ denote relations that are the satisfying assignments of sets of Horn clauses, \emph{i.e.} clauses  with at most one positive (unnegated) literal. For each Horn clause, the negated literals correspond to a subset of the left wires, while the only positive literal corresponds to the right wire to which they are connected. Intuitively,  matrix diagrams can be thought of simply as matrices with Boolean coefficients: each row encodes the negative literals of a Horn clause through its coefficients. An arbitrary $n\times m$ matrix $A$ can be represented by a matrix diagram with $m$ wires on the left and $n$ wires on the right---the left ports can be interpreted as the columns and the right ports as the rows of $A$ (\emph{cf.}~\cite[\S 3.2]{zanasi2015interacting}). The $j$-th wire on the left is connected to the $i$-th wire on the right whenever $A_{ij}=1$; when $A_{ij}=0$ they remain disconnected. For example,
the matrix $A = 
\begin{psmallmatrix}
0 & 1 & 1\\
0 & 1 & 0 \end{psmallmatrix}
$ 
can be depicted as $\scalebox{0.8}{\input{matrix-ex.tikz}}$. 
Conversely, given a diagram of this form, we can recover the corresponding matrix by looking at which ports are connected to which.
The equational theory $\axBARel$ is complete for the monoidal subcategory generated by $\copyy, \discard, \unit, \conjj$.
\begin{proposition}[Matrix completeness]\label{prop:matrix-completeness}
Let $c,d \colon m\to n$ be two diagrams formed only of the generators $\copyy, \discard, \unit, \conjj$. We have $\intBProf{c} = \intBProf{d}$ iff $c=d$ in $\axBARel$.
\end{proposition}
We can extract a unique matrix from a pre-normal form diagram, which we will call its \emph{representing matrix}. Moreover,  
Proposition~\ref{prop:matrix-completeness} allows us to simplify our reasoning in the proof of completeness: in light of this result, we will assume  -- wherever it is convenient -- that we can identify any two diagrams formed of $\copyy, \discard, \unit, \conjj$ that represent the same matrix. This often means that we can be as lax as necessary when identifying diagrams modulo the (co)associativity of $\copyy$ or $\conjj$ and the (co)unitality of $(\copyy, \discard)$ and $(\conjj, \unit)$. In particular, given a Boolean matrix $A$, we can define a unique (up to equality in $\axBARel$) pre-normal form diagram with representing matrix $A$. 

%
%
%

Pre-normal forms are closely related to certain sets of Boolean clauses which we now introduce. 
We fix a finite set of variables $\Var = \{ x_1, \dots, x_m \}$. By analogy with the corresponding terminology for Boolean formulas (\emph{cf.} Section~\ref{sec:semantics}), we call a subset of $\Var$ a \emph{\negClause{}}.
In the rest of this section, a \negClause{} $\{x_{i_1}, \cdots, x_{i_k}\}$ is intended to represent the diagram $\input{cup-white-k.tikz}$ whose semantics is given by $\{x_{i_1}, \cdots, x_{i_k}\,|\, x_{i_1} \land \cdots \land x_{i_k} \leq 0\}$ which is equal to the set of satisfying assignments of $\lnot x_{i_1} \lor \cdots \lor \lnot x_{i_k}$\footnote{Strictly speaking, because all literals are negative, we should write $\{\lnot x_{i_1}, \cdots, \lnot x_{i_k}\}$ for the corresponding clause. However, to save the reader a flurry of negations, we will simply write it as the set $\{x_{i_1}, \cdots, x_{i_k}\}$.}.
%
%
For conciseness, we sometimes write an assignment $\valS \colon \Var \to \Bool$ as 
an element of $\Bool^m$, whose $i$-th element is $\valS(i+1)$, for $i = 0, \dots, m-1$. Given a \negClause{} $\varphi=\{x_{i_1}, \cdots, x_{i_k}\}$, we say that an assignment $(b_1,\dots, b_m)$ satisfies $\varphi$ if $\lnot b_{i_1}\lor \cdots\lor \lnot b_{i_k} = 1$ (equivalently, if $b_{i_1}\land \cdots \land b_{i_k} = 0$). We say that it satisfies a set of \negClause{} if it satisfies all \negClause{} in the set simultaneously. 

\noindent\textbf{From pre-normal form to \negClauses{}
} 
Pre-normal form diagrams matter because one can derive a set of \negClauses{} from them, which is closely related with their semantics (Proposition~\ref{prop:one-one-clause-NF-dcset}). 
Given some pre-normal form diagram $d \colon m \to 0$, we define $\sysIneq{}{d}$ to be the 
set of \negClauses{} over $m$ variables as follows.  
Each row $\rowR$ of the representing matrix $A$ of $d$ generates a \negClause{} $\varphi$ such that $x_i \in \varphi$ if and only if $\rowR[i] = 1$, and $\sysIneq{}{d}$ is the set of all \negClauses{} generated by the rows of $A$. 
The diagram $d$ and the set of \negClauses{} $\sysIneq{}{d}$ are semantically equivalent in the following sense. 
\begin{proposition}
\label{prop:gen-ineq-is-semantics}
Let $d \colon m \to 0$ be a pre-normal form diagram. Then $\intBProf{d}$ is exactly the set of all satisfying assignments of $\sysIneq{}{d}$, the le
set of \negClauses{} generated by $d$.
\end{proposition}
\begin{example}
The diagram $d = \scalebox{0.8}{\input{eg-pre-NF-give-ineq-3.tikz}}$ is in pre-normal form and $\sysIneq{}{d} = \{\{x_1, x_2\}, \{x_1, x_3\}, \{\}\}$. 
$\intBProf{d} = \empset$ is the set of all satisfying assignments of $\sysIneq{}{d}$. 
\end{example}

\noindent\textbf{From \negClauses{} to pre-normal form} 
Conversely, given a set of \negClauses{} $\Phi$ over $\Var$, we can define a diagram $\diag{\Phi} \colon m \to 0$ in pre-normal form. 
Assume that $|\Phi| = n$, and we list the \negClauses{} in $\Phi$ using lexicographical order as $\varphi_1, \dots, \varphi_n$ (for instance, $\{ x_1, x_3 \}$ appears before $\{ x_2 \}$).
Then $\diag{\Phi}$ is defined as the pre-normal form diagram with representing $n\times m$ Boolean matrix $A$ given by $A_{ij} = 1$ precisely if $x_j \in \varphi_i$. 
%
%
Again, $\Phi$ and $\diag{\Phi}$ are semantically equivalent in the following sense. 
\begin{lemma}
\label{lem:diag-sol-of-sys}
$\intBProf{\diag{\Phi}}$ is exactly the set of satisfying assignments of $\Phi$. 
\end{lemma}

\begin{example}
\label{eg:sys-ineq-to-diag}
Let $\Var = \{ x_1, x_2, x_3 \}$, and consider the following two sets of \negClauses{}: 
$\Gamma = \{\{x_1\}, \{x_2, x_3\}\}$ and $\Phi = \{\{x_1\}, \{x_1, x_2\}, \{x_2, x_3\}\}$. 
They have the same set of satisfying assignments $A = \{ 000, 010, 001 \}$; their associated pre-normal form diagrams are 
$\diag{\Gamma} = \scalebox{0.7}{\input{eg-pre-NF-give-ineq-1.tikz}}$, 
$\diag{\Phi} = \scalebox{0.7}{\input{eg-pre-NF-give-ineq-2.tikz}}$, 
which verify $\intBProf{\diag{\Gamma}} = \intBProf{\diag{\Phi}} = A$. 
Moreover, $\Phi$ contains the `redundant' \negClause{} $\{x_1, x_2\}$ because the clause $\{x_1\}$ already implies that $\{x_1, x_2\}$ is satisfied (if $\lnot x_1 = 1$, then $\lnot x_1 \lor \lnot x_2 = 1$), while $\Gamma$ contains no such redundancies. Soon we shall see that $\Gamma$ is the smallest set of \negClauses{} whose set of satisfying assignments is $A$, and $\diag{\Gamma}$ is in normal form.

\end{example}

\subsection{Normal form and completeness.} 
Recall that for completeness we want to select one representative diagram -- in normal form --  for each semantic object. 
%
%
%
We use sets of \negClauses{} as an intermediate step and choose such representative to be the \emph{minimal set of clauses} among those with a given set of satisfying assignments: a set of \negClauses{} $\Phi$ is \emph{minimal} if dropping any \negClause{} from $\Phi$ returns a set $\Phi'$ whose set of satisfying assignments is a proper superset of those of $\Phi$. Minimal sets of \negClauses{} are unique, thus every given set of satisfying assignment for some set of \negClauses{} has a least set of \negClauses{}.
\begin{lemma}
\label{lem:min-system-eq}
If two minimal 
sets of \negClauses{} $\Sigma$ and $\Gamma$ have the same satisfying assignments, then $\Sigma = \Gamma$.
\end{lemma}
\begin{lemma}
\label{lem:dcset-solution-ineq}
A subset of $\Bool^{\Var}$ is downward closed if and only if it is the satisfying assignment of some set of \negClauses{} over $\Var$. 
\end{lemma}
Lemmas~\ref{lem:min-system-eq}-\ref{lem:dcset-solution-ineq} imply the existence of a unique minimal set of \negClauses{} for each semantic object. 
\begin{lemma}
\label{lem:iso-dcsubset-minclause}
For every downward closed subset $A$ of $\Bool^{\Var}$ there is a unique minimal
set of \negClauses{} whose solution set is $A$.
\end{lemma}
%
The normal form diagrams are defined as the diagrammatic counterpart of 
minimal sets of \negClauses{}. 
\begin{definition}
\label{def:normal-form}
Suppose $d \colon m \to 0$ is a pre-normal form diagram with representing matrix $A$. 
We say $d$ is of \emph{normal form} if there are no two distinct rows $i, j$ of $A$ such that $A_{ik} \leq A_{jk}$ for all $k = 1, \dots, m$. 
\end{definition}

\begin{wrapfigure}{r}{0pt}
    \input{eg-normal-form-6.tikz}
\end{wrapfigure}

Note that swapping two rows in the representing matrix $A$ of diagram $d$ does not change $\sysIneq{}{d}$ (the set of \negClauses{} generated by $d$), so what we really care about is normal form diagrams up to some permutation of the rows of their representing matrices. 
Let $d$ and $e$ be two normal form diagrams with representing $n \times m$ matrices $A$ and $B$, respectively. $d$ and $e$ are \emph{equivalent up to commutativity} if $A$ is $B$ with its rows permuted: for all $i \in \{ 1, \dots, n \}$, there exists some $j \in \{ 1, \dots, n \}$ such that $A_{i k} = B_{j k}$ for all $k = 1, \dots, m$, and vice versa. 

\begin{example}
\label{eg:normal-form}
Consider the pre-normal form diagrams above right. 
$c_1$ is not of normal form because it represents the matrix $\begin{psmallmatrix} 
1 & 1 & 0 \\
1 & 1 & 1 
\end{psmallmatrix}$. 
%
%
$c_2$ and $c_3$ are both of normal form, representing matrices 
$\begin{psmallmatrix}
1 & 1 & 0\\
0 & 1 & 1 \end{psmallmatrix}$ 
and
$\begin{psmallmatrix}
0 & 1 & 1\\
1 & 1 & 0 \end{psmallmatrix}$ 
respectively, thus equivalent up to commutativity. 
\end{example}
\begin{proposition}
\label{prop:one-one-clause-NF-dcset}
There is a 1-1 correspondence between the following three sets: 
\begin{enumerate}
    \item  downward closed subsets of $\Bool^m$;
    \item minimal sets of \negClauses{} over $\Var = \{ x_1, \dots, x_m \}$;
    \item normal form diagrams of type $m \to 0$, modulo equivalence up to commutativity.
\end{enumerate}
\end{proposition}
\begin{remark}\label{rem:full-completeness}
The proposition above implies a stronger result plain completeness : $\axBARel$ is \emph{fully} complete for the given semantics, as every monotone relation between tuples of Booleans is in the image of the interpretation functor $\intBProf{\cdot}$.
\end{remark}
Having established the connection between semantic objects and normal form diagrams, we show that indeed every $\catSyn$-morphism is equal to one in normal form. 
Crucially we need the following lemma, which allows us to remove redundant clauses from a diagram in pre-normal form to obtain one in normal form:
if some \negClause{} $\phi$ is strictly larger than another $\psi$, then one can drop $\phi$ from $\Phi$ and the resulting set of \negClauses{} has the same set of satisfying assignments with the original $\Phi$. 
\begin{lemma}
\label{lem:eliminate-redun-ineq}
For arbitrary natural numbers $m, n$, 
\scalebox{0.8}{\input{remove-redundant-ineq-1.tikz}}
\end{lemma}
\begin{example}
Recall the two diagrams $\diag{\Gamma}$ and $\diag{\Phi}$ from Example~\ref{eg:sys-ineq-to-diag}, and we note that $\diag{\Phi}$ is in normal form while $\diag{\Gamma}$ is not. Yet we can apply Lemma~\ref{lem:eliminate-redun-ineq} (to the grey block) to rewrite $\diag{\Gamma}$ (on the lhs) into the normal form diagram $\diag{\Phi}$ (on the rhs): 
$\scalebox{0.6}{\input{eg-eliminate-redun-ineq-2.tikz}}$
\end{example}
%


\begin{lemma}[Normal form]
\label{lem:normal-form}
Every diagram $d \colon m \to 0$ is equal to a diagram in normal form.
\end{lemma}
%

%
%
Finally we are ready to prove the completeness part of Theorem~\ref{thm:sound-complete-monotone-rel-FBA}. 
We start with the following completeness statement (involving only equalities), which follows from Proposition~\ref{prop:one-one-clause-NF-dcset} and Lemma~\ref{lem:normal-form}. 
\vspace{-10pt}
\begin{lemma}
\label{lem:complete-eq}
For any diagrams $c, d \colon m \to n$, if $\intBProf{c} = \intBProf{d}$, then $c = d$.
\end{lemma}
\noindent We can then derive completeness from the following line of reasoning. 
\vspace{-10pt}
\begin{proof*}{Proof of Theorem~\ref{thm:sound-complete-monotone-rel-FBA}.} 
We assume that $\intBProf{c} \subseteq \intBProf{d}$. By Lemma~\ref{lem:syn-plus-leq}\emph{(ii)}, 
$\intBProf{\scalebox{0.8}{\input{plus-c-d.tikz}}} = \intBProf{c}$. By Lemma~\ref{lem:complete-eq}, this implies $\scalebox{0.8}{\input{plus-c-d.tikz}} = c$, thus $c \leq d$ according to Lemma~\ref{lem:syn-plus-leq} \emph{(i)}.
\end{proof*}


\section{Conclusion}\label{sec:conclusion}

An alternative approach to the problem of Boolean satisfiability is to consider the symmetric monoidal category of \emph{all} relations between tuples of Booleans. A close cousin of this category -- that of spans between finite sets with cardinality $2^n$ -- was axiomatised in~\cite{comfort2020zx}. Adding a single axiom (to enforce transitivity of the order on $\Bool$) should be sufficient to obtain a complete axiomatisation of the corresponding category of relations. The resulting calculus would certainly be expressive enough to encode SAT instances and derive their (un)satisfiability equationally. However, the equational theory is more complex (reflecting the larger semantic domain), with two distinct Frobenius structures. Moreover, the monotonicity of our semantics, while seemingly limiting at first, allows us to \emph{topologise} certain algebraic aspects of SAT: negation is represented as a change of the direction of wires, an essential feature of our calculus. These choices result in a tighter correspondence with formulas in CNF and a simpler equational theory, with close links to existing SAT-solving algorithms, as we have begun to investigate in Section~\ref{sec:sat-solving}. We hope to explore these connections further in future work, as well as the possibility of devising new heuristics for SAT-solving guided by the equational theory. 

On the logic programming side, 
the current paper extends work from \cite{gu2021functorial}, in which two of the authors of the current paper propose a functorial semantics for logic programs. 
Using the completeness result of the paper, we are able to go further to obtain a diagrammatic calculus to compute Herbrand semantics of logic programs, and show two programs are equivalent by proving the equivalence of two diagrams. 
There are two main directions of further exploration. 
The first is to extend the current framework to include more logical components, for example to consider logic programming admitting variables and/or negation. 
The second is to invent a diagrammatic calculus for probabilistic and weighted logic programming. 

\section*{Acknowledgement} Thanks to Guillaume Boisseau and Cole Comfort for helpful discussion on the completeness proof and  related~work. RP and FZ acknowledge support from EPSRC grant EP/V002376/1.

\bibliographystyle{entics}
\bibliography{refs}

\appendix
\section{Representing monotone functions as monotone relations}
\label{sec:mono-func-mono-rel}
Let us first restate the definition of a monotone relation representing a monotone function.
\begin{definition}
\label{def:represent-monotone-func}
Let $\lr{X, \leq_X}, \lr{Y, \leq_Y}$ be two partially ordered sets. We say a monotone function $f \colon X \to Y$ is \emph{represented by} a monotone relation $R \subseteq X \times Y$ if for arbitrary $x\in X$ and $y \in Y$, $(x, y) \in R$ if and only if $f(x) \leq y$. 
\end{definition}

\begin{proposition}
\label{prop:rel-rep-faithful}
If two monotone functions $f, g$ are represented by the same monotone relation $R$, then $f = g$. 
\end{proposition}
\begin{proof}
We prove by contradiction. Suppose $f \neq g$ and they are represented by the same relation $R$, then there exists $x\in X$ such that $f(x) \neq g(x)$. Yet by the definition of $R$, both $(x, f(x))$ and $(x, g(x))$ are in $R$. We claim that this entails $f(x) = g(x)$.

Since $R$ represents $g$, $(x, f(x)) \in R$ implies that there exists $x' \in X$ such that $x \leq x'$ and $g(x') \leq f(x)$. Note that $g$ is monotone, so $g(x) \leq g(x') \leq f(x)$. Similarly, $R$ represents $f$ and $f$ is monotone imply that $f(x) \leq g(x)$. Now that $\lr{Y, \leq_Y}$ is a partial order means that $\leq_Y$ is antisymmetric, so $f(x) = g(x)$, which contradicts the assumption that $f(x)  \neq g(x)$. Therefore $f = g$.
\end{proof}

\begin{proposition}
\label{prop:rep-comp}
The notion `being represented by' is closed under both sequential and parallel compositions:
\begin{enumerate}
    \item If $f \colon X \to Y$ and $g \colon Y \to Z$ are represented by $R\subseteq X \times Y$ and $S \subseteq Y \times Z$, respectively, then $f \then g \colon X \to Z$ is represented by $R \relThen S$. 
    \item If $f \colon X \to Y$ and $g \colon U \to V$ are represented by $R \subseteq X \times Y$ and $S \subseteq U \times V$, respectively, then $f \times g \colon X \times U \to Y \times V$ is represented by $R \times S \subseteq (X \times U) \times (Y \times V)$.
\end{enumerate}
\end{proposition}
\begin{proof}
For (i), we show that for arbitrary $x \in X$ and $z\in Z$, $(x, z) \in R \relThen S$ if and only if $(f\then g) (x) \leq z$. For the `only if' direction, suppose $(x, z) \in R \relThen S$, then there exists $y \in Y$ such that $(x, y) \in R$ and $(y, z) \in S$. Since $R$ and $S$ represent $f$ and $g$ respectively, we know that $f(x) \leq y$ and $g(y) \leq z$, so by the monotonicity of $g$, $(f\then g)(x) = g(f(x)) \leq g(y) \leq z$ holds. For the `if' direction, suppose $(f\then g)(x) \leq z$, then $f(x) \leq f(x)$ implies that $(x, f(x)) \in R$. Also, $g(f(x)) \leq z$ implies that $(f(x), z) \in S$. So $(x, z) \in R\relThen S$. 

(ii) is immediate from the fact that both $f \times g$ and $R \times S$ are defined pointwise. 
\end{proof}

\section{Some useful diagrammatic (in)equations}\label{sec:derived-laws}
This section contains some (in)equations in $\axBARel$. While some of them appear in several related monoidal theories (\emph{e.g.}, the generalised bimonoid equations of Proposition~\ref{prop:gen-bialg}), the equation derived in Proposition~\ref{prop:loop} is a distinguishing feature of $\axBARel$, as far as we know. 
\begin{proposition}
\label{prop:monoid-unit-dif-color}
\input{unit-monoid-dif-color.tikz}
\end{proposition}
\begin{proof}
We only prove the first two equations, and the other two follow from a similar proof by reflecting all the diagrams in the proof.
For the first equation, 
    \[
        \input{unit-monoid-dif-color-1.tikz}
    \]
For the second equation, 
    \[
        \input{unit-monoid-dif-color-2.tikz}
    \]
\end{proof}
\begin{proposition}
\label{prop:loop}
\begin{equation*}
    \input{lem-loop-1.tikz}
\end{equation*}
\end{proposition}
\begin{proof}
Crucially we need (A14) for the second inequality.
    \begin{equation}
        \input{lem-loop-2.tikz}
    \end{equation}
\end{proof}
\begin{corollary}\label{cor:loop-removal}
\begin{equation*}
\input{not-special.tikz}
\end{equation*}
\end{corollary}
\begin{proof}
\input{not-special-proof.tikz}
\end{proof}
\begin{proposition}\label{prop:conj-dist-disj}
The following distributivity equation holds:
\begin{equation}
    \input{lem-conj-dist-disj.tikz}
\end{equation}
\end{proposition}
\begin{proof}
\begin{equation}
    \input{proof-conj-dist-disj-2.tikz}
\end{equation}
\end{proof}
\begin{proposition}\label{prop:idempotent-equalities}
The converse inequalities to (D2), (D5), and (D10) are derivable.
\end{proposition}
\begin{proof}
For (D2) we can derive the converse inequality as follows:
\input{bcomult-wmult-equal-id.tikz}
The other two are similar. For (D5) replace (D3) by (D7) and apply unitality (A6) of $\cocopy,\codiscard$ and counitality (A6) of $\copyy,\discard$. For (D10) replace (D3) by (D7) and apply unitality (A6) of $\cocopy,\codiscard$ and counitality (B2) of $\coconj,\counit$. 
\end{proof}

\begin{proposition}
\label{prop:gen-bialg}
The following general bimonoid equation holds for all natural numbers $k, n$.
\begin{equation}
\label{eq:gen-bialg}
    \input{gen-bialg.tikz}
\end{equation}
\end{proposition}
\begin{proof}
For readability, in the rest of the proof we will not use the dashed wires but instead use $\vdots$ and natural numbers to intuitively express that `there are $k$ duplicates of certain pattern'. 

We prove by induction on $n$, and leave $k$ arbitrary. 
When $n = 0$ and $n = 1$, we have
\begin{equation*}
    \input{gen-bialg-proof-1.tikz}
\end{equation*}
When $n = 2$, we make an induction on $k$. 
The base case for $k = 0, 1$ are exactly the same as the previous proof for $n = 0, 1$. 
For $k = 2$, this is exactly the bimonoid axiom. 
Now suppose it holds for $k \geq 2$, and we consider $k + 1$:
\begin{equation*}
    \input{gen-bialg-proof-3.tikz}
\end{equation*}

We continue with the induction proof on $n$. Suppose equation \eqref{eq:gen-bialg} holds for $n$, then for $n + 1$, 
\begin{equation*}
    \input{gen-bialg-proof-2.tikz} 
\end{equation*}
\end{proof}

\begin{proposition}
\label{prop:gen-bialg-mix}
The following bimonoid equation between $\coconj$ and $\cocopy$ (resp. $\conjj$ and $\copyy$) holds for arbitrary natural numbers $k, n$ 
\begin{equation}
\label{eq:gen-bialg-BW}
    \input{gen-bialg-BW.tikz}
\end{equation}
\end{proposition}
\begin{proof}
The proof is the same as that for Proposition~\ref{prop:gen-bialg}. Note that there only properties of $(\copyy, \cocopy)$ forms a bimonoid is used, which also holds for $(\copyy, \conjj)$ and $(\cocopy, \coconj)$, see the B and C axioms in Figure~\eqref{fig:ineq-axiom}. 
\end{proof}


\section{Matrices, diagrammatically}
\label{sec:matrix-diagrams}
In this section we summarise the diagrammatic representation of Boolean matrices as $\catSata$ diagrams. A similar representation appears in many related theories (\emph{e.g.}, in Graphical Linear Algebra, see~\cite[Section 3.2]{zanasi2015interacting}), but the connection between matrices and the semantics of $\catSata$ diagrams as monotone relations is not as immediate and deserves further explanations. 
\begin{definition}
\label{def:matrix-diagram}
We say that $d$ is a \emph{matrix diagram} if it factorises as follows:
\begin{equation}
    \input{matrix-diagram.tikz}
\end{equation}
\end{definition}
Matrix diagrams $m\to n$ denote relations that are the satisfying assignments of sets of Horn clauses, \emph{i.e.} clauses  with at most one positive (unnegated) literal. For each Horn clause, the negated literals correspond to a subset of the left wires, while the only positive literal corresponds to the right wire to which they are connected. For this reason,  matrix diagrams can also be thought of simply as matrices with Boolean coefficients: each row encodes the negative literals of a Horn clause through its coefficients. An arbitrary $n\times m$ matrix $A$ can be represented by a matrix diagram with $m$ wires on the left and $n$ wires on the right---the left ports can be interpreted as the columns and the right ports as the rows of $A$. The $j$-th wire on the left is connected to the $i$-th wire on the right whenever $A_{ij}=1$; when $A_{ij}=0$ they remain disconnected. For example,
\[\text{the Boolean matrix } A = 
\begin{pmatrix}
0 & 1 & 1\\
0 & 1 & 0
\end{pmatrix} \text{ can be depicted as } \input{matrix-ex.tikz}\]
Conversely, given a diagram of this form, we can recover the corresponding matrix by looking at which ports are connected to which.
\begin{definition}
\label{def:canonical-mat-diag}
A \emph{canonical matrix diagram} $d \colon m \to n$ is a diagram of the following form:
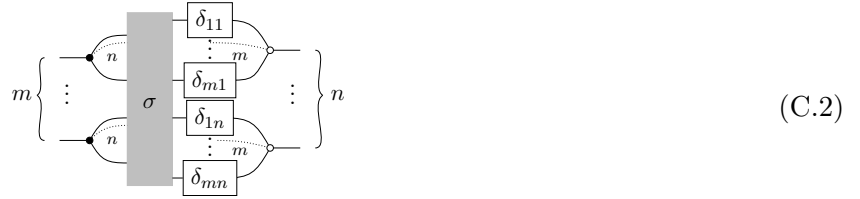
\begin{equation}
    \scalebox{0.8}{\input{mat-diag-canonical-1.tikz}}
\end{equation}
where each $\matEntry{i}{j}$ is either $\begin{tikzpicture}
	\begin{pgfonlayer}{nodelayer}
		\node [style=none] (0) at (-0.75, 0) {};
		\node [style=none] (1) at (0, 0) {};
		\node [style=none] (2) at (0.75, 0) {};
	\end{pgfonlayer}
	\begin{pgfonlayer}{edgelayer}
		\draw (0.center) to (1.center);
		\draw (1.center) to (2.center);
	\end{pgfonlayer}
\end{tikzpicture}
$ or $\begin{tikzpicture}
	\begin{pgfonlayer}{nodelayer}
		\node [style=none] (0) at (-1, 0) {};
		\node [style=black dot] (1) at (-0.25, 0) {};
		\node [style=empty dot] (2) at (0.25, 0) {};
		\node [style=none] (3) at (1, 0) {};
	\end{pgfonlayer}
	\begin{pgfonlayer}{edgelayer}
		\draw (0.center) to (1);
		\draw (2) to (3.center);
	\end{pgfonlayer}
\end{tikzpicture}
$, for $i = 0, \dots, m-1$ and $j = 0, \dots, n-1$; $\sigma$ contains suitably many $\swap$ that move the $j$-th output wire of the $i$-th $\copyy$ to the $i$-th input wire of the $j$-th $\conjj$.
\end{definition}
The equational theory $\axBARel$ is complete for the monoidal subcategory generated by $\copyy, \discard, \unit, \conjj$.
\begin{proposition}[Matrix completeness]
Let $c,d \colon m\to n$ be two diagrams formed only of the generators $\copyy, \discard, \unit, \conjj$. We have $\intBProf{c} = \intBProf{d}$ iff $c=d$ in $\axBARel$.
\end{proposition}
\begin{proof}
The symmetric monoidal category of matrices with Boolean coefficients and the direct sum as monoidal product has a well-known complete axiomatisation: the theory of a commutative and idempotent bimonoid. This result can be found in~\cite[Proposition~3.9]{zanasi2015interacting}: therein the third author of this paper gives a complete axiomatisation of the symmetric monoidal category of matrices over a principal ideal domain (again with the direct sum as monoidal product). However the proof in~\cite{zanasi2015interacting} does not make any use of additive inverses and can be adapted without changes to the case of an arbitrary semiring, like the Booleans.

In $\axBARel$ the axioms defining a commutative and idempotent bimonoid are given in Fig.~\ref{fig:ineq-axiom} by equations (A1)-(A13). Note however that we care about matrices formed using $\copyy, \discard, \unit, \conjj$ here, so we need to check that the same equations hold with $\cocopy$ replaced by $\conjj$ and $\codiscard$ replaced by $\unit$. For this, we only need to show that we can strengthen (D2) to an equality (for idempotency). The converse inequality can be derived as follows:
\[\input{bcomult-wmult-equal-id.tikz}\]
This shows that the symmetric monoidal subcategory of $\catSata$ generated by $\copyy, \discard, \unit, \conjj$ embeds faithfully into the symmetric monoidal category of Boolean matrices. 

The desired completeness result follows because the symmetric monoidal category of Boolean matrices embeds faithfully in our target category of monotone relations and that $\intBProf{\cdot}$ (or rather, the restriction of $\intBProf{\cdot}$ to the symmetric monoidal subcategory generated by $\copyy, \discard, \unit, \conjj$) factors through this embedding. To see this, notice that an $n\times m$ Boolean matrix $A$ defines a monotone function $f\colon \Bool^m\to \Bool^n$ whose $j$-th component is given by $f_j(x_1,\dots,x_m) =  x_{j_1}\land \cdots \land  x_{j_k}$ where $\{j_1,\dots, j_k\}$ is the set of indices in the $j$-th row of $A$ which are equal to $1$. Clearly, any two matrices that define the same monotone function are equal. In turn, we can embed monotone functions faithfully into monotone relations, as explained in Proposition~\ref{prop:rel-rep-faithful}.
\end{proof}
\begin{corollary}
\label{cor:matrix-diag-comp}
Any diagram formed only of the generators $\copyy, \discard, \unit, \conjj$ is equal to a canonical matrix diagram. 
\end{corollary}
\begin{example}
\label{eg:canonical-matrix-diag}
The following matrix diagram is equivalent to a canonical matrix diagram on the right hand side:
\begin{equation}
    \scalebox{0.8}{\input{mat-diag-canonical-eg-3.tikz}}
\end{equation}
In particular, the grey diagram is the $\sigma$ in Definition~\ref{def:canonical-mat-diag}; $\matEntry{2}{1} = \matEntry{1}{3} = \matEntry{2}{3} = \begin{tikzpicture}
	\begin{pgfonlayer}{nodelayer}
		\node [style=none] (0) at (-1, 0) {};
		\node [style=black dot] (1) at (-0.25, 0) {};
		\node [style=empty dot] (2) at (0.25, 0) {};
		\node [style=none] (3) at (1, 0) {};
	\end{pgfonlayer}
	\begin{pgfonlayer}{edgelayer}
		\draw (0.center) to (1);
		\draw (2) to (3.center);
	\end{pgfonlayer}
\end{tikzpicture}
$, while the rest $\matEntry{i}{j}$'s are all $\begin{tikzpicture}
	\begin{pgfonlayer}{nodelayer}
		\node [style=none] (0) at (-0.75, 0) {};
		\node [style=none] (1) at (0, 0) {};
		\node [style=none] (2) at (0.75, 0) {};
	\end{pgfonlayer}
	\begin{pgfonlayer}{edgelayer}
		\draw (0.center) to (1.center);
		\draw (1.center) to (2.center);
	\end{pgfonlayer}
\end{tikzpicture}
$.
\end{example}
Proposition~\ref{prop:matrix-completeness} will allow us to simplify our reasoning in the proof of completeness. In light of this result, we will assume  -- wherever it is convenient -- that we can always rewrite any diagram formed of $\copyy, \discard, \unit, \conjj$ to its equivalent matrix canonical form. This often means that we can be as lax as necessary when identifying diagrams modulo the (co)associativity of $\copyy$ or $\conjj$ and the (co)unitality of $(\copyy, \discard)$ and $(\conjj, \unit)$.

\begin{remark}\label{rem:three-bimonoids}
The equational theory of Fig.~\ref{fig:ineq-axiom} contains three commutative and idempotent bimonoids: $(\copyy, \discard, \unit, \conjj)$ as we have already explained, but also $(\copyy, \discard, \codiscard, \cocopy)$ and $(\coconj, \counit, \codiscard, \cocopy)$. Indeed, one can check that the relevant axioms also hold for them: they are easy consequences of (A1)-(A13), (C1)-(C6), and (D1)-(D12). We use the encoding of Boolean matrices using $(\coconj, \counit, \codiscard, \cocopy)$ once in Section~\ref{sec:sat}, when translating SAT instances to a diagrams.
\end{remark}

Moreover, the canonical form of matrix diagrams provides a simple and formal definition of the intuitive notion of connected wires.
\begin{definition}
\label{def:connect-canonical-mat-diag}
Let $d$ be a canonical matrix diagram as in Definition~\ref{def:canonical-mat-diag}. 
The $i$-th input wire and the $k$-th output wire are \emph{connected} if $\matEntry{i}{k} = \begin{tikzpicture}
	\begin{pgfonlayer}{nodelayer}
		\node [style=none] (0) at (-0.75, 0) {};
		\node [style=none] (1) at (0, 0) {};
		\node [style=none] (2) at (0.75, 0) {};
	\end{pgfonlayer}
	\begin{pgfonlayer}{edgelayer}
		\draw (0.center) to (1.center);
		\draw (1.center) to (2.center);
	\end{pgfonlayer}
\end{tikzpicture}
$.
The $i$-th and the $j$-th input wires are \emph{connected} if there exists $k \in \{ 0, \dots, n-1 \}$ such that $\matEntry{i}{k} = \matEntry{j}{k} =  \begin{tikzpicture}
	\begin{pgfonlayer}{nodelayer}
		\node [style=none] (0) at (-0.75, 0) {};
		\node [style=none] (1) at (0, 0) {};
		\node [style=none] (2) at (0.75, 0) {};
	\end{pgfonlayer}
	\begin{pgfonlayer}{edgelayer}
		\draw (0.center) to (1.center);
		\draw (1.center) to (2.center);
	\end{pgfonlayer}
\end{tikzpicture}
$. 
\end{definition}
Intuitively, two wires are connected if there is a path between these two wires. This notion of connectivity turns out to be crucial in the proof of completeness (Theorem~\ref{thm:sound-complete-monotone-rel-FBA}). 
\begin{example}
    Consider the diagram in Example~\ref{eg:canonical-matrix-diag}. The second input wire and the first output wire are intuitively `disconnected', and this corresponds to $\matEntry{2 1} = \begin{tikzpicture}
	\begin{pgfonlayer}{nodelayer}
		\node [style=none] (0) at (-1, 0) {};
		\node [style=black dot] (1) at (-0.25, 0) {};
		\node [style=empty dot] (2) at (0.25, 0) {};
		\node [style=none] (3) at (1, 0) {};
	\end{pgfonlayer}
	\begin{pgfonlayer}{edgelayer}
		\draw (0.center) to (1);
		\draw (2) to (3.center);
	\end{pgfonlayer}
\end{tikzpicture}
$. 
\end{example}

\end{document}